\pdfoutput=1
\documentclass[pdftex,pra,aps,twocolumn,twoside,nofootinbib,showkeys,showpacs,10pt]{revtex4-1}

\usepackage[T1]{fontenc}
\fontencoding{T1}  
\usepackage[utf8]{inputenc}

\usepackage{enumitem}
\usepackage{tikz}
\usepackage{amsfonts}
\usepackage{amsmath,amssymb,amsthm}
\usepackage{graphicx}
\usepackage{fancyhdr}
\usepackage[breakable]{tcolorbox}
\usepackage{bbm}
\usepackage{url}
\usepackage{mathtools}
\usepackage{braket}
\usepackage{upgreek }
\usepackage{dsfont}

\usepackage[plain]{fancyref} 
\usepackage{algorithm}
\usepackage{algpseudocode}

\usepackage{hyperref}

\theoremstyle{plain}               
\newtheorem{thm}{Theorem}
\newtheorem{lem}{Lemma}[section]
\newtheorem{cor}{Corollary}[section]
\newtheorem{prop}{Proposition}[section]
\newtheorem{defn}{Definition}[section]
\newtheorem{example}{Example}[section]
\theoremstyle{remark}

\newcommand*{\fancyrefthmlabelprefix}{thm}
\newcommand*{\fancyreflemlabelprefix}{lem}
\newcommand*{\fancyrefcorlabelprefix}{cor}
\newcommand*{\fancyrefdefilabelprefix}{defi}
\frefformat{plain}{\fancyreflemlabelprefix}{lemma\fancyrefdefaultspacing#1}
\Frefformat{plain}{\fancyreflemlabelprefix}{Lemma\fancyrefdefaultspacing#1}
\frefformat{plain}{\fancyrefthmlabelprefix}{theorem\fancyrefdefaultspacing#1}
\Frefformat{plain}{\fancyrefthmlabelprefix}{Theorem\fancyrefdefaultspacing#1}
\frefformat{plain}{\fancyrefcorlabelprefix}{corollary\fancyrefdefaultspacing#1}
\Frefformat{plain}{\fancyrefcorlabelprefix}{Corollary\fancyrefdefaultspacing#1}
\frefformat{plain}{\fancyrefdefilabelprefix}{definition\fancyrefdefaultspacing#1}
\Frefformat{plain}{\fancyrefdefilabelprefix}{Definition\fancyrefdefaultspacing#1}
\newcommand*{\fancyrefalglabelprefix}{alg}
\newcommand*{\frefalgname}{algorithm}
\newcommand*{\Frefalgname}{Algorithm}
\frefformat{plain}{\fancyrefalglabelprefix}{%
  \frefalgname\fancyrefdefaultspacing#1%
}%
\Frefformat{plain}{\fancyrefalglabelprefix}{%
  \Frefalgname\fancyrefdefaultspacing#1%
}%

\def\beq{\begin{equation}}
\def\eeq{\end{equation}}
\def\bq{\begin{quote}}
\def\eq{\end{quote}}
\def\ben{\begin{enumerate}}
\def\een{\end{enumerate}}
\def\bit{\begin{itemize}}
\def\eit{\end{itemize}}

\def\lb{\left(}
\def\rb{\right)}

\def\r|{\right|}

\newcommand\C{\mathbbm{C}}

\newcommand\Z{\mathbbm{Z}}
\newcommand\R{\mathbbm{R}}

\newcommand\M{\mathcal{M}}

\newcommand{\n}{\mathcal{N}}

\newcommand{\ketbra}[2]{|#1\rangle\langle #2|}
\newcommand{\tr}[1]{\operatorname{tr}\lb#1\rb}
\newcommand{\one}{\mathds{1}}

\newcommand{\highlight}[1]{\textcolor{black}{#1}}

\newcommand{\cO}{\mathcal{O}}

\newcommand\be{\begin{equation}}
\newcommand\ee{\end{equation}}
\newcommand{\ab}{\mathbf{a},\mathbf{b}}

\begin{document}
\title{Efficient classical simulation and benchmarking of quantum processes in the Weyl basis}

\author{Daniel Stilck Fran\c{c}a$^1$, Sergii Strelchuk$^2$, Micha{\l} Studzi\'nski$^3$}
\affiliation{
	$^1$ 	QMATH, Department of Mathematical Sciences, University of Copenhagen, Universitetsparken 5, 2100
	Copenhagen, Denmark \\
	$^2$ DAMTP, Centre for Mathematical Sciences, University of Cambridge, Cambridge~CB30WA, UK \\
	$^3$  Institute of Theoretical Physics and Astrophysics, National Quantum Information Centre, Faculty of Mathematics, Physics and Informatics, University of Gda{\'n}sk, Wita Stwosza 63, 80-308 Gda{\'n}sk, Poland}

\begin{abstract}
One of the crucial steps in building a scalable quantum computer is to identify the noise sources which lead to errors in the process of quantum evolution. Different implementations come with multiple hardware-dependent sources of noise and decoherence making the problem of their detection manyfoldly more complex. We develop a randomized benchmarking algorithm which uses Weyl unitaries to efficiently identify and learn a mixture of error models which occur during the computation. We provide an efficiently computable estimate of the overhead required to compute expectation values on outputs of the noisy circuit relying only on locality of the interactions and no further assumptions on the circuit structure. The overhead decreases with the noise rate and this enables us to compute analytic noise bounds that imply efficient classical simulability.
We apply our methods to ansatz circuits that appear in the Variational Quantum Eigensolver and establish an upper bound on classical simulation complexity as a function of noise, identifying regimes when they become classically efficiently simulatable.
\end{abstract}

\maketitle

Any device designed to take advantage of quantum-mechanical features is susceptible to noise which accompanies the underlying physical realization. 
Reliable error correction is one of the major challenges which prevents us from building scalable hardware. The resource overhead to implement even the simplest error-correcting schemes that underpin fault-tolerant computation are currently prohibitively costly. This motivated a flurry of research into quantum algorithms~\cite{Preskill2018quantumcomputingin} that work on quantum computers with small, but non-negligible, error rates and take advantage of quantum information processing protocols before the era of universal, error-corrected quantum computers. One of the key challenges is to precisely understand and characterize the noise and decoherence effects affecting these devices and to investigate how the noise affects the complexity of their classical simulation.

Successful error mitigation relies on correctly identifying the parameters of the underlying error models. The latter are constructed by employing gate-dependent benchmarking suites~\cite{wallman2014randomized, helsen2018representations,Onorati_2019} which aim to characterise the singular sources of noise. One way to use the acquired knowledge about the noise in quantum computing scenarios is to introduce a quantitative measure such as quantum volume~\cite{cross2019validating, blume2019volumetric}. The latter requires to compute the largest achievable depth of a model (random) quantum circuit that can be executed on quantum hardware by estimating the number of `heavy' output strings it generates. While this may give some insight into the reliability of quantum computer, it has several apparent limitations. 

First, it utilizes Haar-random circuits and thus does not provide means to understand hardware performance when implementing a given quantum circuit. 
Second, estimating quantum volume has an unfavourable scaling with a system size because the underlying heavy output generation problem scales exponentially with a number of qubits~\cite{aaronson2016complexity}.

In our work, we introduce an approach to randomized benchmarking and classical simulation of quantum circuits that relies on Weyl unitaries. It enables us to identify a number of error models and demonstrates favourable scaling with the system size which works both for qubit and higher dimensional systems. In particular, we can identify and detect mixtures of channels such as depolarizing and dephasing channels affecting the implementation of a given gate. Having access to noiseless Clifford gates, we can also identify the parameters of other noise models including over-rotations.
 Moreover, our protocol is robust to the so-called state preparation and measurement (SPAM) errors and scalable under a natural assumption that the noise is local. 

Second, we find a surprising connection between benchmarking protocols in the Weyl basis and the ability to simulate outputs of quantum circuits on a classical computer. For a {\it given} quantum circuit with the established noise profile, we provide an analytic bound on the sufficient number of samples required to classically estimate the circuit output up to a given precision using a Feynman path-like algorithm. Thus, our methods can be used to give an upper bound on the computational power of the noisy quantum device with a clear operational interpretation: we can establish a non-trivial computable bound on the gate noise that needs to be added to each gate in the circuit in order to render it classically efficiently simulatable.  All preexisting methods of efficient classical simulation must necessarily assume a particular structure of gateset, whereas our simulation algorithm does not rely on these assumptions. Moveover, our tools do not depend on the geometry of the circuit, which provides an advantage over the state of the art tensor network methods which classically simulate quantum circuits by contracting a tensor network with cost exponential in the treewidth of the graph induced by the circuit~\cite{Markov_2008}. 

The algorithm scales particularly well for estimating the expectation value of Pauli observables on the output of local, low-depth circuits. As this is the main subroutine in quantum algorithms for near-term devices~\cite{kandala2017hardware, 
wang2019accelerated, asaad2016independent}, our tools can be readily used to bound the classical simulation complexity of a wide range of quantum devices used for example in the VQE regime.

{\it{Weyl unitaries}}.
Our protocol makes use of \emph{Weyl-Heisenberg} unitaries $\{W_{(a,b)}\}_{a,b=0}^{d-1}$, which present the generalization of the Pauli matrices higher dimensions. They are defined as $W_{(a,b)}=Z^aX^b$, where $X$ is the shift unitary $X\in U(d)$, which acts on the computational basis mapping $\ket{j}\mapsto\ket{j+1\mod{d}}$, and $Z\in U(d)$ is the phase unitary mapping $\ket{j}\mapsto e^{i\frac{2j\pi}{d}}\ket{j}$, $j=0,\ldots, d-1$. These unitaries have a number of useful properties: they are orthogonal with respect to the Hilbert-Schmidt scalar product, they form an orthogonal basis for ${\cal M}_{d}(\mathbb{C})$ and they are a (projective) representation of $\Z_d\times\Z_d$. 
When describing a system consisting of $n$ qudits, then $ W_{(a_1,b_1)}\otimes W_{(a_2,b_2)}\otimes\cdots\otimes W_{(a_n,b_n)}$ is a basis of the space $\M_{d^n}$, we will usually denote these matrices by $W_{(\textbf{a},\textbf{b})}$, where $(\textbf{a},\textbf{b})\in\lb\Z_d\rb^{2n}$. Then an arbitrary operator $X:\M_d\to\M_d$ can be expressed as a $d^n\times d^n$ matrix with entries $X((\textbf{a},\textbf{b}),(\textbf{c},\textbf{d}))=d^{-n} \tr{ W_{(\textbf{a},\textbf{b})}^\dagger X (W_{(\textbf{c},\textbf{d})})}$, we will use a short notation $\big<X\big>^{(\textbf{a},\textbf{b})}_{(\textbf{c},\textbf{d})}$ for this representation.

{\it Classical simulation of noisy quantum circuits.} 
Working with Weyl unitaries brings forth the importance of using compact yet rich representation space for studying quantum processes. In this setting we go beyond standard benchmarking and make use of the information about noise in the circuit to bound its classical simulation complexity. 

Every state $\rho$  can be represented as a vector w.r.t. this basis by setting $\rho(\mathbf{a},\mathbf{b})=d^{-n/2}\tr{W_{(\mathbf{a},\mathbf{b})}^\dagger\rho}$. The same holds for observables, and quantum channels.
Consider a (noisy) circuit  ${\cal C}_{\cal B} = \n^{(N)}\circ \cdots \circ \n^{(1)}$, where the $\n^{(i)}:\M_{d^n}\to\M_{d^n}$ are quantum channels, acting on a product input $\rho= \rho_1\otimes\cdots\otimes\rho_n$. For a given observable $E$, we can classically simulate ${\cal C}_{\cal B}$ if we can estimate $\tr{\sigma E}$ classically up to an additive error $\epsilon>0$, where $\sigma={\cal C}_B(\rho)$ (this is also known as the {\it weak} simulation). 

We will make use of the $\ell_p$-norms of matrices w.r.t. to this basis, denoted by $\|\cdot\|_{\ell_p}$(which is just the $\ell_p$ norm of $\rho$ when regarded as a vector).
This is not to be confused with the usual Schatten norms.

Let $\n(\ab)$ be the $(\ab)$-th column of the channel $\n$. 
Then $\|\n(\ab)\|_{\ell_1}=\sum\limits_{(\mathbf{a}',\mathbf{b}')\in\Z_d^{2n}}\left| \big<\n\big>^{(\textbf{a},\textbf{b})}_{(\textbf{a}',\textbf{b}')}\right| $, 
and $\|\n\|_{{\ell_1}\to {\ell_1}}$ is the maximum over all $(\ab)$.
Our algorithm is based on $\ell_1$ sampling of matrices and vectors. We will assume that given quantum channels $\n^{(k)}$ in our (noisy) circuit, one can get the sample corresponding to its action on an input basis element $(\ab)$. That is, we assume we can sample from $
p_k(\mathbf{a},\mathbf{b}|\mathbf{a}',\mathbf{b}')=|\big<\n^{(k)}\big>^{(\textbf{a},\textbf{b})}_{(\textbf{a}',\textbf{b}')}|/\|\n^{(k)}(\ab)\|_{\ell_1}
$
and from 
$p_0(\ab)=|\rho(\ab)|/\|\rho\|_{\ell_1}.
$
As highlighted in~\cite{Rall2019}, it may be convenient to work in the Heisenberg picture, which requires replacing $\rho$ with $E$ in the  equation above. For simplicity, we present the algorithm in the Schr\"odinger picture.

\begin{tcolorbox}
[ title={Circuit sampling algorithm}]
\emph{Input:} noisy quantum circuit specified by quantum channels $\n^{(1)},\ldots,\n^{(N)}$, initial quantum state $\rho$ and observable $E$.\\
\emph{Output:} complex number $x$ s.t. $\mathbb{E}(x)=\tr{E\bigcirc_{k=1}^N\n^{(k)}(\rho)}$.
\begin{enumerate}
\itemsep-0.2em 
\item Sample $(\mathbf{a}_0,\mathbf{b}_0)$ from the distribution $p_0$.
\item For $k=1,\ldots,n$:
 Sample $(\mathbf{a}_{k},\mathbf{b}_k)$ from $p_k(\mathbf{a}_{k+1},\mathbf{b}_{k+1}|\mathbf{a}_{k},\mathbf{b}_k)$
\item Output $x$ given by
\begin{align*}
x&=\operatorname{sign}(\rho(\mathbf{a}_0,\mathbf{b}_0))\|\rho\|_{\ell_1}E(\mathbf{a}_{n},\mathbf{b}_n)\times\\
&\prod\limits_{k=1}^{N}\|\n^{(k)}(\mathbf{a}_{k},\mathbf{b}_k)\|_{\ell_1} \operatorname{sign}(\big<\n^{(k)}\big>^{(\textbf{a}_k,\textbf{b}_k)}_{(\textbf{a}_{k+1},\textbf{b}_{k+1})})
\end{align*}
\end{enumerate}
\end{tcolorbox}
Here $\operatorname{sign}(\cdot)$ function denotes the phase pre-factor.
Variations of this algorithm have recently and independently been discussed in other contexts~\cite{1903.04483,Rall2019}.
The following theorem proves the correctness of the algorithm by showing that it samples from the true distribution:
\begin{thm}\label{samplingthm}
The output of the circuit sampling algorithm satisfies $\mathbb{E}(x)=\operatorname{tr}\left(\sigma E\right)$. Taking $\cO\left(\frac{1}{\epsilon^2}M_B\log\lb\frac{1}{\delta}\rb\right)$ many samples, where
\begin{align}\label{equ:upperboundoutput}
M_B=\lb\|\rho\|_{\ell_1}\|E\|_{\ell_\infty}\prod\limits_{k=1}^{N}\|\n^{(k)}\|_{{\ell_1}\to{\ell_1}} \rb^2
\end{align}
suffices to guarantee that with probability at least $1-\delta$ an empirical average $\bar{x}$ of the samples satisfies $|\bar{x}-\operatorname{tr}\left(\sigma E\right)|\leq\epsilon$.
\end{thm}
\begin{proof}
The probability of a fixed sequence $((\mathbf{a}_0,\mathbf{b}_0),\ldots,(\mathbf{a}_N,\mathbf{b}_N))$ is given by:
\begin{align}\label{equ:proboutcome}
\frac{|\hat{\rho}(\mathbf{a}_0,\mathbf{b}_0)|}{\|\hat{\rho}\|_{\ell_1}} \prod\limits_{k=1}^{N}\frac{|\big<\n^{(k)}\big>^{(\textbf{a}_k,\textbf{b}_k)}_{(\textbf{a}_{k+1},\textbf{b}_{k+1})}|}{\|\n^{(k)}(\textbf{a}_k,\textbf{b}_k)\|_{\ell_1}}.
\end{align}
Forming a product of the probability for a given sequence in~\eqref{equ:proboutcome} and the corresponding output of the algorithm we get its expectation value by summing over all possible outcome sequences:
\begin{align}\label{equ:matrixmultiplication}
\tr{\sigma E}=\qquad\smashoperator{\sum_{\substack{(\mathbf{a}_{1},\mathbf{b}_{1}),\ldots,(\mathbf{a}_{N},\mathbf{b}_{N})}}}E(\mathbf{a}_{N},\mathbf{b}_{N}) Z_N\dots Z_1\rho(\mathbf{a}_{0},\mathbf{b}_{0}),
\end{align}
where $Z_i=\big<\n^{(i)}\big>^{(\textbf{a}_{i-1},\textbf{b}_{i-1})}_{(\textbf{a}_{i},\textbf{b}_{i})}$.
The bound on the necessary number of samples follows from Hoeffding's inequality~\cite{hoeffding1994probability} after observing that the absolute of the output in the algorithm is at most $\|\hat{\rho}\|_{\ell_1}\|\hat{E}\|_\infty\prod\limits_{k=1}^{n}\|\n^{(k)}\|_{{\ell_1}\to{\ell_1}}.
$
\end{proof}
The proof only relies on the linearity of the evolution and not any property of the basis or underlying maps and vectors. Therefore, it can be easily re-expressed in the Heisenberg picture, i.e. by replacing sampling from $\rho$ with sampling from $E$ and $(\n^{(k)})^*$ in reverse order. This is useful when $\|E\|_{\ell_1}$ is smaller than $\|\rho\|_{\ell_1}$ or when we estimate averages of strings of Pauli operators.
This sampling routine is remarkably versatile: we extend our sampling algorithm to the case of quantum circuits that made up of quantum channels of the form $e^{t\mathcal{L}}$, where $\mathcal{L}$ is a Lindbladian. In addition, it also applies to unitary evolutions defined by Hamiltonian dynamics. 

Our circuit sampling algorithm extends that of~\cite{pashayan2015estimating} in several ways. Firstly, we show that one can use the results of the randomized benchmarking experiments to bound the complexity of a given noisy device in the Weyl basis. Second, it works for noisy quantum circuits in continuous and discrete time and evolutions in the Heisenberg picture. We present a range of bases for these as well as the Lindbladian case in Section VI and VII
 of the  \highlight {Supplemental Material (SM)} respectively. 

{\it Weyl randomized benchmarking (WRB)}.  An important feature of the Weyl operators is that many practically relevant noise models, such as (local) dephasing or (local) depolarizing channels, are diagonal in the Weyl operator basis. We will refer to such channels as \emph{Weyl diagonal channels} and denote them as $\mathcal{T}$. It turns out that for $d=2$ they coincide with mixed Pauli channels and for $d>2$ correspond to convex combinations of conjugations with the Weyl operators~\cite[Chapter 4]{Watrous2018}.

Thus, given their ubiquity and the fact that randomized compiling protocols can even bring arbitrary noise to this form~\cite{Wallman2016}, the goal of our protocol will be to learn the parameters of a Weyl-diagonal channel that models the noise affecting a unitary through randomized benchmarking protocols~\cite{Knill_2008,Magesan_2012,Magesan_2011}.

When implementing a known unitary $U$ acting on $n$ qudits, the resulting transformation, due to noise effects, is described by the quantum channel $\mathcal{T}\circ \mathcal{U}$, where $\mathcal{U}$ is the channel which corresponds to conjugation with $U$ followed by $\mathcal{T}$. 
Our goal is to learn the parameters of a Weyl diagonal channel $\mathcal{T}$, i.e. its diagonal elements in the Weyl basis.

To implement the protocol we make the following assumptions about the noise: a) one can implement Weyl unitaries with negligible error, and b) successive implementations of $\mathcal{U}$ are followed by the same error channel $\mathcal{T}$. We discuss how to relax the first assumption in Section II
 of the  \highlight {SM}.  The protocol consists of the following steps:
 
\begin{tcolorbox}
[ title={Weyl randomized benchmarking (WRB) protocol}]
\emph{Input:} $(\textbf{a},\textbf{b})\in\lb\Z_d\rb^{2n}$ corresponding to the diagonal element we wish to learn and a sequence length $m$. Initial state $\rho$ and POVM element $E$ on $n$ qudits.\\ 
\emph{Output:} complex number $y$.
\begin{enumerate}
\itemsep-0.2em  
\item Draw a random $(\textbf{a}_0,\textbf{b}_0)\in\lb\Z_d\rb^{2n}$, apply $W_{(\textbf{a}_0,\textbf{b}_0)}$ followed by a sequence $\bar W = (W_{(\textbf{a}_1,\textbf{b}_1)},\ldots,W_{(\textbf{a}_m,\textbf{b}_m)})$ of uniformly random local Weyl unitaries on the $n$ qudits interspersed with the (noisy) unitary $U$.
\item Apply $\bar W^\dagger$. 
\item Measure the state with a POVM $\{E,\one-E\}$. 
\item When $E$ is measured, output $y=$ $\chi_{(\textbf{a},\textbf{b})}(\textbf{a}_0,\textbf{b}_0)=$ $\text{exp}(i\frac{2\pi}{d}\langle(\textbf{b},-\textbf{a}),(\textbf{a}_0,\textbf{b}_0)\rangle)$
 Else, output $y=0$.
\end{enumerate}
\end{tcolorbox}
The function $\chi_{(\textbf{a},\textbf{b})}$ is the character of a representation of the group $(\Z_d)^{2n}$ and it ensures that we project the initial state to $W_{(\textbf{a},\textbf{b})}$.
More precisely, it follows from standard representation theory that for any operator $Y\in\M_{d^n}$:
\begin{equation}
\begin{split}
   (1/d^{2n})\smashoperator{\sum_{\substack{(\textbf{a}_0,\textbf{b}_0)\in (\Z)^{2n}}}}&\chi_{(\textbf{a},\textbf{b})}(\textbf{a}_0,\textbf{b}_0)W_{(\textbf{a}_0,\textbf{b}_0)}YW_{(\textbf{a}_0,\textbf{b}_0)}^{\dagger}=\\
   &(1/d^n)\tr{W_{(\textbf{a},\textbf{b})}^\dagger Y}W_{(\textbf{a},\textbf{b})}.
   \end{split}
\end{equation}
This forces the expectation value for a fixed sequence length $m$ to be given by 
\begin{align}\label{equ:expectation_diagonal}
  \frac{1}{d^{2n}}\tr{ W_{(\textbf{a},\textbf{b})}^{\dagger} \mathcal{T}\circ \mathcal{U}\lb W_{(\textbf{a},\textbf{b})}\rb}^mS,
\end{align}
where $S=\tr{W_{(\textbf{a},\textbf{b})}^\dagger\rho}\tr{E W_{(\textbf{a},\textbf{b})}}$.
(See Section I 
 for a gentle introduction of character randomized benchmarking and Section II
 of  \highlight {SM} for the justification of the above).
Our protocol is related to character randomized benchmarking of~\cite{helsen2018representations}, with the distinction that we wish to determine the noise affecting a specific unitary assuming that the noise affecting Weyl operators is negligible. The protocol does not significantly depend on the particular choice of $E$ and $\rho$, as long as $\tr{W_{(\textbf{a},\textbf{b})}^\dagger\rho}\tr{E W_{(\textbf{a},\textbf{b})}}\sim d^n$. This is because we will later perform an exponential fit of Eq.~\eqref{equ:expectation_diagonal} to a curve of the form  $a\times q^m$, and the magnitude of $a$ is determined by $E$ and $\rho$. Canonical choices to achieve this scaling would be to pick $E$ as the projector onto the +1 eigenspace of $W_{(\textbf{a},\textbf{b})}$ and $\rho$ as one of its eigenvalues because for this choice we have $\tr{E W_{(\textbf{a},\textbf{b})}}\geq d^{n-1}$ and measuring these POVMs only requires product measurements. We refer to Section II.D of the \highlight {SM}  for more details.

By selecting different sequence lengths and performing an exponential fitting one gets an estimate of the diagonal in the Weyl basis of $\mathcal{T}\circ \mathcal{U}$. The maximal sequence length is determined by the spectral gap $\lambda$ of the quantum channel $\mathcal{T}$. For symmetric (i.e. $\mathcal{T}=\mathcal{T}^*$) Weyl-diagonal channels, this reduces to $1-\lambda_2$, where $\lambda_2$ is second largest eigenvalue, and is a natural measure of the noisiness of the channel. The parameter $\lambda^{-1}$ should be thought of as the depth at which the noise clearly manifests itself, as $\lambda$ can be seen as a generalized error probability of the channel. For instance, for a depolarizing channel with depolarizing probability $p$, $\lambda=p$, and we expect to see errors at depth $p^{-1}$.  We then have:
\begin{thm}\label{benchmarkingthm}
Let $\mathcal{T}:\M_{d^n}\to\M_{d^n}$ be a symmetric (i.e. $\mathcal{T}=\mathcal{T}^*$) Weyl-diagonal channel, $\mathcal{U}:\M_{d^n}\to\M_{d^n}$ a unitary channel, $\epsilon,\delta>0$ be given error parameters, and  let $\lambda$ be the spectral gap of $\mathcal{T}\circ  \mathcal{U} $.
Then we can find an estimate $\hat{\mu}(\textbf{a},\textbf{b})$ of 
$\mu(\textbf{a},\textbf{b})=\big<\mathcal{T}\circ  \mathcal{U} \big>^{(\textbf{a},\textbf{b})}_{(\textbf{a},\textbf{b})}$
satisfying
$
\left|\mu(\textbf{a},\textbf{b})-\hat{\mu}(\textbf{a},\textbf{b})\right|\leq \cO\left(\epsilon|1-|\mu(\textbf{a},\textbf{b})|^2|\right)
$
with probability at least $\delta$ by performing 
$
M=\cO(\epsilon^{-2}\log(\delta^{-1}\log(1-\lambda)^{-1}))
$
randomized benchmarking experiments each containing at most $M_{\max}=\cO\lb\lambda^{-1}\rb$ gates in the sequence.
\end{thm}
Proof of the Theorem is located in Section IV
 of the  \highlight {SM} and is based on extending the results of~\cite{Harper_2019_statistical} to our setting. We note that in the case of qubits the underlying channels are always symmetric and many other relevant examples such as dephasing or depolarizing channel belong to this class.
The knowledge about $\mu(\textbf{a},\textbf{b})$ and the fact that $\mathcal{T}$ is Weyl diagonal makes it sufficient to estimate $\mu(\textbf{a},\textbf{b})$ to estimate the noise parameters because in this case: $
    \mu(\textbf{a},\textbf{b})=
    \big<\mathcal{U}\big>^{(\textbf{a},\textbf{b})}_{(\textbf{a},\textbf{b})}\big<\mathcal{T}\big>^{(\textbf{a},\textbf{b})}_{(\textbf{a},\textbf{b})}$.
We thus completely characterize $\mathcal{T}$ as long as the diagonal of the unitary is nonzero. It turns out that a simple enhancement of the above protocol with a suitable noiseless Clifford gate allows one to analyze noise models with {\it any} off-diagonal contributions. Consider an element $\big<\mathcal{T}\big>^{(\textbf{a}_1,\textbf{b}_1)}_{(\textbf{a}_2,\textbf{b}_2)}$ of $\mathcal{T}$ in the Weyl basis which we want to estimate and assume we can implement a noiseless Clifford that acts as $CW_{(\textbf{a}_1,\textbf{b}_1)}C^\dagger=e^{i\phi}W_{(\textbf{a}_2,\textbf{b}_2)}$. Such unitary always exists and can be easily identified (as long as none of the Weyl operators involved is the identity).
Applying $C$ after the target unitary gate in the randomized benchmarking experiment gives access to the desired off-diagonal entry:
    \begin{equation}\label{equ:Ccliffordtrick}
\mu(\textbf{a}_1,\textbf{b}_1)= \big<\mathcal{C}\circ \mathcal{T}\circ  \mathcal{U} \big>_{(\textbf{a}_1,\textbf{b}_1)}^{(\textbf{a}_1,\textbf{b}_1)}=e^{-i\phi}\big<\mathcal{T}\circ  \mathcal{U}\big>_{(\textbf{a}_1,\textbf{b}_1)}^{(\textbf{a}_2,\textbf{b}_2)}.
\end{equation}
The output of the algorithm is an estimate of $\big<\mathcal{T}\big>^{(\textbf{a}_1,\textbf{b}_1)}_{(\textbf{a}_2,\textbf{b}_2)}$ as per Theorem~\ref{benchmarkingthm}. This method can also be used to learn any number of matrix elements of $\mathcal{T}\circ\mathcal{U}$ by conjugating it with Pauli matrices, interpolating between a constant number of learnable noise parameters analysed in~\cite{1806.02048} and full process tomography~\cite{Kimmel_2014,Roth2018}. 

To estimate Weyl-diagonal channels one requires $\cO(d^{2n})$ parameters, which remains practical only for small systems. 
However, assuming locality of the noise it is possible to learn it efficiently.
For example, suppose that that the unitary $U$ is the product of $2$-qudit gates followed by Weyl-diagonal noise acting on the same qudits. Let the noise on the first qudit be completely characterized by the diagonals w.r.t. Weyl operators with $(\textbf{a},\textbf{b})=(a_1,a_2,0,
\ldots,0,b_1,b_2,0,\ldots,0)$. This gives a total of $\cO(nd^4)$ parameters to learn, rendering the protocol efficient. We discuss the `Clifford trick' of Eq.~\ref{equ:Ccliffordtrick} as well as the extension to more complex local noise models in Sections II.C
and III of the  \highlight {SM} respectively.

Theorem~\ref{benchmarkingthm} extends the results of~\cite{flammia2019efficient,Harper2019} in two distinct ways. First, our techniques are not qubit-specific and work for systems of arbitrary dimension. Secondly, we are able to naturally incorporate gate-dependent noise (as long as we as make assumptions about how it affects the Weyl operators). Thus, we relax the assumption whereby the Weyl operators all being affected by the same, known Weyl diagonal (noise) channel. 

{\bf Applications.} {\it Local quantum circuits}. When the quantum noise channels in the circuit are local and the initial state and observable are product, then the complexity of our sampling algorithm scales polynomially. To achieve this  one requires efficient estimation of the transition probabilities and sampling and/or access to entries of either the state or observable in the basis. We restrict our discussion to the Weyl basis, but the argument works for any product basis.

Suppose that each quantum channel acts on at most $m=\cO(1)$ qubits.  
If we have a product basis and $\n^{(k)}$ is local, 
then $\big<\n^{(k)}\big>^{(\textbf{a}_1,\textbf{b}_1)}_{(\textbf{a}_2,\textbf{b}_2)}=0$ if the strings differ outside of the support of $\n^{(k)}$. This is because the action of $\n^{(k)}$ does not change that element of the string and, thus, the output of $(\textbf{a}_1,\textbf{b}_1)$ remains orthogonal to the other string. Thus, given some $(\textbf{a}_1,\textbf{b}_1)$ as input, it suffices to only compute $\big<\n^{(k)}\big>^{(\textbf{a}_1,\textbf{b}_1)}_{(\textbf{a}_2,\textbf{b}_2)}$ for $(\textbf{a}_2,\textbf{b}_2)$ that coincides with the input on the support of $\n^{(k)}$ to get the elements with nonzero probability under $p_k(\cdot|(\textbf{a}_1,\textbf{b}_1))$. As there are only $d^{2m}=\cO(1)$ many of these, computing the associated quantities such as the normalization and signs can be done in polynomial time, resulting in the efficient routine to which produces samples.  It remains to estimate $M_B$ to determine a bound on the required number of samples. Note that
\begin{align}\label{equ:multnegtensor}
\|\n^{(1)}\otimes\n^{(2)}\|_{\ell_1\to\ell_1}=\|\n^{(1)}\|_{\ell_1\to\ell_1}\|\n^{(2)}\|_{\ell_1\to\ell_1}
\end{align}
in the case of a product basis. Moreover, the $\ell_1\to\ell_1$ norm is submultiplicative as a matrix norm induced by a vector norm, i.e. $\|\n^{(1)}\circ \n^{(2)}\|_{\ell_1\to \ell_1}\leq\|\n^{(1)}\|_{\ell_1\to\ell_1}\|\n^{(2)}\|_{\ell_1\to \ell_1}$. 

The above properties are used for get an estimate of $M_B$. We use of the multiplicativity of the $\ell_1\to\ell_1$ norm given by equation~\eqref{equ:multnegtensor} for subsequences of the circuit consisting of quantum channels that do not overlap. Each individual $\ell_1\to\ell_1$ norm can  be computed efficiently and the multiplicativity implies that the overall $\ell_1\to\ell_1$ norm of this sequence of operations is just the product of each one. Then, whenever two quantum channels have a nontrivial overlap, we may use the submultiplicativity of the norms and  computing the norm for subsequences consisting of non-overlapping quantum channels. In short, we see that if $\n^{(N)}\circ \n^{(N-1)}\circ\cdots\circ \n^{(1)}$ is a sequence of local, noisy gates that describe the circuit, then the number of samples is at most as in~\eqref{equ:upperboundoutput} in Theorem~\ref{samplingthm}.
Thus we can sample efficiently in the Weyl basis from circuits consisting only of local quantum channels. 

The Weyl basis has many advantages over the phase space basis when simulating algorithms on near-term quantum hardware and studying the effects of noise. First, Clifford gates represented in this basis do not increase the sample complexity of the algorithm: they act as signed permutations in the Weyl basis and, thus, $\|\mathcal{C}\|_{\ell_1\to\ell_1}=1$ for any Clifford gate $\mathcal{C}$.
Secondly, if the initial state is product and the target observable is local or is a Pauli string, then we can also achieve that $\|E\|_{\ell_1}\|\rho\|_{\ell_\infty}=\cO(1)$ by simulating the evolution in the Heisenberg picture (see Section V of the SM)

{\it Simulating VQE ansatze.}
The simplicity of representation of Clifford gates as well as Pauli observables makes this method suitable for classically simulating quantum circuits that appear in the VQE algorithm. We apply out tools to the problem of solving MaxCut on a graph with $n$ vertices using the VQE algorithm~\cite{moll2017quantum}. The problem is encoded in the ground state of the Hamiltonian $H = \sum_{1=i<j}^n w_{ij}\sigma^z_i\otimes \sigma^z_j$, where $w_{ij}\in \mathbb{R}$. The ansatz circuit used in this case for the state preparation has the form $|\psi({\boldsymbol{\theta}})\rangle =\left[U({\boldsymbol \theta})U_{ent}\right]^D|\psi({\mathbf 0})\rangle$, ${\boldsymbol{\theta}} = \{\theta_{i,k}\}_{i,k}$, $1\le i\le n, 1\le k\le D$ where the $k$-th application of parametrized unitary is given by $U({\boldsymbol \theta}) = \otimes_{i=1}^n Y(\theta_{i,k})$, $Y(\theta_{i,k}) = \exp{(-i \frac{\theta_{i,k}}{2}\sigma^Y_i)}$, and $U_{ent}=\otimes_{i=1}^{n/2-1}CNOT_{2i,2i+1}$. The VQE algorithm works by iteratively preparing states $|\psi({\boldsymbol{\theta}})\rangle$ which are the approximations of the ground state of $H$, where $\boldsymbol \theta$ in each iteration are determined by a suitable classical optimization algorithm. 

Now assume we performed the WRB protocol and estimated that each $CNOT$ gate in the ansatz experiences a two-local depolarizing noise $p_C$ and the single-qubit rotations suffer from one local depolarizing noise with rate $p_Y$. We assume this rate to be independent of $\theta$ for simplicity. $CNOT$ gate is a Clifford gate, but the $Y(\theta)$ are in general non-Clifford gates. Taking into account the noise, the process of state evolution in the Weyl basis can be represented as $\rho({\boldsymbol \theta})_W  = \n^{(D)}_{ent}\circ \n^{(N)}_Y\circ \cdots \circ \n^{(1)}_{ent}\circ \n^{(1)}_Y|\psi({\boldsymbol 0})_W\rangle$. Note that both the noise and the gates are unital. To sample from the circuit we turn to Theorem~\ref{samplingthm}. Clifford unitaries act as signed permutations, making $\ell_1\to \ell_1$ norm equal to $1$. Bounds are improved if we also incorporate the noise affecting the Clifford gates in the rotations: we assume that the CNOTs are noiseless, while each pair of rotations is preceded by $2$-local depolarizing noise $\mathcal{T}_{p_C}$ and proceeded by $1$-local depolarizing noise $\mathcal{S}_{p_Y}$.
Viewing this evolution in the Pauli basis, a simple computation shows that for $\mathcal{T}_{p_C}\circ Y(\theta)\otimes Y(\theta') \circ\lb \mathcal{S}_{p_Y}\otimes \mathcal{S}_{p_Y}\rb$ we have:
\begin{equation}
\begin{split}
&\|\mathcal{T}_{p_C}\circ Y(\theta)\otimes Y(\theta') \circ\lb \mathcal{S}_{p_Y}\otimes \mathcal{S}_{p_Y}\rb\|_{\ell_1\to \ell_1}=\\&\max\{1, p_y^2p_C\phi(\theta)\phi(\theta'),p_yp_C\phi(\theta),p_yp_C\phi(\theta')\},
\end{split}
\end{equation}
where $\phi(\theta)=|\cos(\theta)|+|\sin(\theta)|\leq\sqrt{2}$.

To sample the energy values with respect to $H$ which contains at most $n^2$ two-body Pauli observables with accuracy $\epsilon$ we require at most $M = \cO(\epsilon^{-2}n^2\lb p_Cp_Y^22\rb^{2nD})$ samples for $p_Y\geq2^{-1/2}$. In particular, the algorithm is efficient whenever $p_Cp_Y^2<1/2$. Thus, as long as $p_Cp_Y^2\leq 1+\tfrac{\log(nD)}{nD}$ we have that the number of samples scales as $M=\cO(n^4D^2\epsilon^{-2})$, making the algorithm efficient because $\lb p_Cp_Y^22\rb^{2nD}\leq n^2D^2$.


{\it Acknowledgements.}
D.S.F. was supported by VILLUM FONDEN via the QMATH Centre of Excellence under Grant No. 10059 and the European Research Council (Grant agreement No. 818761).
 S.S. acknowledges support from the QuantERA ERA-NET Cofund in Quantum Technologies implemented within the European Union's Horizon 2020 Programme (QuantAlgo project), and administered through the EPSRC grant EP/R043957/1., the Leverhulme Early Career Fellowship scheme and the Royal Society University Research Fellowship. MS acknowledges support from the grant "Mobilno{\'s}{\'c} Plus IV", 1271/MOB/IV/2015/0 from the Polish Ministry of Science and Higher Education.


\bibliography{biblio}
\newpage
\section*{Supplemental Material}
\section{Character randomized benchmarking}\label{app:charbenchintro}
Any randomized benchmarking protocol is defined with respect to a given, discrete collection of gates called the gateset $\mathbb{G}$. The procedure relies on randomly sampling a sequence of gates from the set $\mathbb{G}$ with the goal of estimating its average fidelity. The sequence of gates is applied to an initial state, which is followed by a global inversion gate. In the ideal situation, when noise is absent, the system returns to the initial configuration. However, this is not the case in practice. In this case, we compute the overlap between the output and the input state by measuring with two-component POVM $\{E,\one-E\}$. Repeating this for a large number of sequences of different lengths $m$ gives us a list of so-called survival probabilities $\{p_m\}_m$.
If the gate set consists of the elements from the Clifford group, and the noise is gate independent, points $\{p_m\}_m$ can be fitted to a single exponential decay curve of the form
\be
\label{fit1}
p(m,E,\rho) \approx A+Bf^m.
\ee
Constants $A,B$ depend on the quality of the state $\rho$ preparation and measurement; and the parameter $f$ informs us how well the gates are implemented.  

In the general case, when the gate set is not the multi-qubit Clifford group the fitting relation~\eqref{fit1} does not hold and it must be replaced by a more general form:
\be
\label{fit2}
p(m,E,\rho) \approx \sum_{\alpha}C_{\alpha}f^m_{\alpha}.
\ee
Parameters $f_{\alpha}$ depend only on the quality of implementation of the gates, prefactors $C_{\alpha}$ depend only on how well the initial state $\rho$ is prepared and measured. 

In particular, when a given gate set forms a group $\mathbb{G}$, we can isolate numbers $f_{\alpha}^m$ given in~\eqref{fit2} using the so-called character randomized benchmarking protocol introduced in~\cite{1806.02048}. According to the original notation by $G$ we denote a unitary gate from the gateset $\mathbb{G}$. By writing $\mathcal{G}(\rho)$ we denote an action $G\rho G^{\dagger}$.
For a general randomized benchmarking procedure over a given group $\mathbb{G}$ introduced in~\cite{Magesan_2012}, we can write numbers $p(m,E,\rho)$ as
\begin{equation}
p(m,E,\rho)=\operatorname{tr}\left[E\left(\frac{1}{|\mathbb{G}|}\sum_{G \in \mathbb{G}}\mathcal{G}^{\dagger}\widetilde{\mathcal{G}} \right)^m(\rho)  \right].
\end{equation}
Here $\widetilde{\mathcal{G}}=\n\circ \mathcal{G}$ is a noisy implementation of the action of $\mathcal{G}(\rho)$.
Applying Schur's lemma~\cite{Roe} we simplify above expression to
\begin{equation}
p(m,E,\rho)=\sum_{\alpha}\operatorname{tr}\left(E\mathcal{P}_{\alpha}(\rho) \right)f_{\alpha}^m, 
\end{equation}
where $\mathcal{P}_{\alpha}$ is the projector onto representation space of irreducible component $\phi_{\alpha}$. Finally denoting the character function of the representation $\phi_{\alpha'}$ as $\chi_{\alpha'}$, assigning to every element $G\in \mathbb{G}$ a complex number we can write down a modified randomized benchmarking protocol (keeping the original sequence):
\begin{tcolorbox}[breakable]
	\emph{Input:} Group $\mathbb{G}$, sequence length $m$, quantum state $\rho$, POVM element $E$ and character $\chi^{\alpha'}$ of $\mathbb{G}$.\\
	\emph{Output:} The estimate of survival probability $k_m$.
	\begin{enumerate}

	\itemsep-0.2em  
		\item Sample $\vec{G}=G_1,\ldots,G_m$ uniformly at random from $\mathbb{G}$.
		\item Sample $\hat{G}$ uniformly at random from $\mathbb{G}$.
		\item Prepare a quantum state $\rho$ and apply the gates $(G_1\hat{G}),G_2,\ldots,G_m$.
		\item Compute the inverse $G_{\operatorname{inv}}=(G_m\cdots G_1)^{\dagger}$ and apply it.
		\item Estimate the weighted survival probability
		\begin{equation}
		\begin{split}
		k_m^{\hat{\alpha}'}(\vec{G},\hat{G})&=\operatorname{tr}(E\mathcal{P}_{\hat{\alpha}'}(\rho))\chi_{\hat{\alpha}'}(\hat{G})\times \\
		&\times\operatorname{tr}\left[Q\widetilde{\mathcal{G}}_{\operatorname{inv}}\widetilde{\mathcal{G}}_m\cdots \widetilde{(\mathcal{G}_1\hat{\mathcal{G}})}(\rho) \right].
		\end{split}
		\end{equation}
		\item Repeat for many $\hat{G}\in \hat{\mathbb{G}}$ and estimate the average 
		\begin{equation}
		k_m^{\hat{\alpha}'}=\mathbb{E}_{\hat{G}}\left(k_m^{\hat{\alpha}'}(\vec{G},\hat{G}) \right).
		\end{equation}
		\item Repeat for many $\vec{G}$ and estimate the average 
		\begin{equation}
		k_m=\mathbb{E}_{\vec{G}}(k_m^{\hat{\alpha}'}(\vec{G})).
		\end{equation}
		\item Repeat for many different $m$.
	\end{enumerate}
\end{tcolorbox}
Using the above and properties of irreducible characters one can rewrite Eq.~\eqref{fit2} as
\begin{equation}
k_m=\operatorname{tr}\left[E\mathcal{P}_{\alpha'}(\rho) \right]f_{\alpha'}^m, 
\end{equation}
so we are able to isolate each parameter.  We then have to choose POVM $E$ and the initial state $\rho$ to maximize $\operatorname{tr}\left[E\mathcal{P}_{\alpha'}(\rho) \right]$ and repeat the procedure for different choices of $\alpha, \widetilde{\alpha}'$.
\section{Randomized Benchmarking for Weyl group}\label{app:RBapp}
We will now specialize statements of the previous section to the case where $G$ is the group $\lb \Z_d\times \Z_d\rb^n$ with the  (projective) representation given by the Weyl operators.
Here we will discuss basic facts related to the Weyl matrices, randomized benchmarking in the Weyl basis and quantum channels that are covariant with respect to Weyl group that are needed for our protocol.

\subsection{Projections onto Weyl matrices}
Let us review basic facts about the Weyl operators on $\M_d$, which were defined in the main text. It is easy to see that they satisfy the relations
\begin{align}
W_{(a_1,b_1)}W_{(a_2,b_2)}=\nu ^{b_1+a_2}W_{(a_1+a_2,b_1+b_2)},
\end{align}
where $\nu=e^{\frac{2\pi}{d}}$ is the $d$-th root of unity. This implies that they form a projective representation of $\Z_d\times \Z_d$ and several other useful relations follow from the formula above, such as
\begin{align}\label{equ:conjugationwithweyl}
W_{(a_1,b_1)}W_{(a_2,b_2)}W_{(a_1,b_1)}^\dagger=\nu^{b_1a_2-a_1b_2}W_{(a_2,b_2)} .  
\end{align}
It also follows that the conjugate action of the Weyl operators in $\M_d$ gives a representation of $\Z_d\times \Z_d$. 

We know from standard representation theory~\cite{Simon_1995} that we can decompose $\M_d$ into irreducible subspaces with respect to this representation of $\Z_d\times \Z_d$. 
Let us now discuss this decomposition into irreducible representations and the projections onto them in the case of Weyl operators. A simple consequence of Eq.~\eqref{equ:conjugationwithweyl} that all the subspaces $\mathbb{W}_{a,b}=\text{span}\{W_{(a,b)}\}$ are invariant under this representation. Thus, each one of these subspaces is an irreducible subspace with respect to this representation. As we have $d^2$ such subspaces and the underlying space has dimension $d^2$, we conclude these are all irreducible subspaces. 

Let us now discuss in more detail the projections onto each one of $\mathbb{W}_{a,b}$, denoted by $\mathcal{P}_{(a,b)}$. It is well-known that the projector onto irreducible subspaces are of the form
\begin{align}\label{equ:projectionirrep}
\mathcal{P}_{(a,b)}(X)=\frac{1}{d^2}\sum_{(f,g)\in\Z_d\times\Z_d}\chi_{(a',b')}(f,g)W_{(f,g)}X W_{(f,g)}^\dagger,
\end{align}
where $\chi_{(a',b')}$ corresponds to the character of some  irreducible representation of $\Z_d\times \Z_d$.
In the case of the group $\Z_d\times \Z_d$, it is well-known that all the characters are of the form 
\begin{align}\label{equ:charactercommutativegroup}
\chi_{(a',b')}(f,g)=\text{exp}\lb \frac{2\pi i}{d}\lb a'f+b'g\rb\rb.
\end{align}
We now show how to pick $a',b'$ in order to get the projection onto $\mathbb{W}_{(a,b)}$ with the help of the conjugation formula in Eq.~\eqref{equ:conjugationwithweyl}. For the projection we have:
\begin{align}
&W_{(a,b)}=\mathcal{P}_{(a,b)}(W_{(a,b)})=\\ &\frac{1}{d^2}\sum_{(f,g)\in\Z_d\times\Z_d}\chi_{(a',b')}(f,g)W_{(f,g)}W_{(a,b)} W_{(f,g)}^\dagger=\\&\frac{W_{(a,b)}}{d^2}\sum_{(f,g)\in\Z_d\times\Z_d}\chi_{(a',b')}(f,g)\nu^{gb-fa}.
\end{align}
Thus, the equation is satisfied whenever $\chi_{(a',b')}(f,g)=\nu^{fa-gb}$. It then easily follows from Eq.~\eqref{equ:charactercommutativegroup} that picking $(a',b')=(b,-a)$ ensures that we get the correct projection.
Moreover, as the set $\{d^{-\frac{1}{2}}W_{(a,b)}\}$ is an orthonormal basis, we also have that 
\begin{align}\label{equ:projectionirrepweyl}
\mathcal{P}_{(a,b)}(X)=d^{-1}\tr{W_{(a,b)}^\dagger X}W_{(a,b)}.
\end{align}
It is easy to see that the same results carry over when we consider tensor products of Weyl operators as representations of $\lb \Z_d\times \Z_d\rb^n$ in $\M_{d^n}$ and we have for $(\mathbf{a},\mathbf{b})\in\lb \Z_d\times \Z_d\rb^n $
\begin{align}
&\mathcal{P}_{(\mathbf{a},\mathbf{b})}(X)=d^{-n}\tr{W_{(\mathbf{a},\mathbf{b})}^\dagger X}W_{(\mathbf{a},\mathbf{b})}=\\
&\frac{1}{d^{2n}}\sum_{(\mathbf{f},\mathbf{g})\in\lb \Z_d\times \Z_d\rb^n}\chi_{(\mathbf{b},-\mathbf{a})}(\mathbf{f},\mathbf{g})W_{(f,g)}X W_{(\mathbf{f},\mathbf{g})}^\dagger,  
\end{align}
where
\begin{align}
\chi_{(\mathbf{b},-\mathbf{a})}=\prod\limits_{i=1}^n\chi_{(b_i,-a_i)}.
\end{align}
In particular, if we pick a uniformly random element of $(\mathbf{f},\mathbf{g})\in\lb \Z_d\times \Z_d\rb^n$ and consider the random linear map
\begin{align}
\mathcal{P}'_{(\mathbf{a},\mathbf{b})}(X)=\chi_{(\mathbf{b},-\mathbf{a})}(\mathbf{f},\mathbf{g})W_{(\mathbf{f},\mathbf{g})}X W_{(\mathbf{f},\mathbf{g})}^\dagger, 
\end{align}
where $(\mathbf{f},\mathbf{g})$ is picked uniformly at random, then
\begin{align}
\mathbb{E}\left[ \mathcal{P}'_{(\mathbf{a},\mathbf{b})}\right]=\mathcal{P}_{(\mathbf{a},\mathbf{b})},
\end{align}
which is the crux of character randomized benchmarking~\cite{Helsen_2017}.

\subsection{Weyl-Covariant channels}
We now turn to studying the structure of quantum channels that are \emph{covariant} with respect to the Weyl unitaries in more detail. Recall that a quantum channel $\n:\M_d\to\M_d$ is covariant with respect to a unitary representation $\phi$ of a group $G$, $\phi:g\mapsto U_g$ if for all $g\in G$ we have:
\begin{align}
\mathcal{U}_g\circ \n\circ  \mathcal{U}_{g^{-1}}=\n.
\end{align}
Recall the characterization of the Weyl-diagonal channels:
\begin{prop}\label{prop:characterizationweyl}
	Let $\n:\M_{d^n}\to\M_{d^n}$ be a quantum channel. Then the following are equivalent:
	\begin{enumerate}
		\item $\n$ is covariant with respect to the representation of $\lb \Z_d\times \Z_d\rb^n$ given by the Weyl operators.
		\item $\n$ is Weyl diagonal.
		\item $\n$ is a mixed Weyl channel, that is, there is a probability distribution $p$ on $\lb \Z_d\times \Z_d\rb^n$, such that
		\begin{align}
		\n(X)=\sum\limits_{(\mathbf{a},\mathbf{b})}p(\mathbf{a},\mathbf{b})W_{(\mathbf{a},\mathbf{b})}X W_{(\mathbf{a},\mathbf{b})}^\dagger,
		\end{align}
		where $X\in \M_{d^n}$.
	\end{enumerate}
\end{prop}
\begin{proof}
	The equivalence between $2$ and $3$ is proved e.g. in~\cite[Chapter 4]{Watrous2018}. 
	The equivalence between $1$ and $2$ follows by a simple direct inspection.
\end{proof}
From this, we get:

\begin{prop}\label{prop:structuretwirling}
	Let $\n:\M_{d^n}\to\M_{d^n}$ be a quantum channel and let $\{W_{(\textbf{a},\textbf{b})}\}$ be the Heisenberg-Weyl matrices, with $(\textbf{a},\textbf{b})\in\lb \Z_d\times\Z_d\rb^n$. Moreover, define $\mu(\textbf{a},\textbf{b})=d^{-n}\operatorname{tr}\left(W_{(\textbf{a},\textbf{b})}^{\dagger}\n(W_{(\textbf{a},\textbf{b}}))\right)$. Then, for $X\in\M_d$:
	\begin{align}
	\frac{1}{d^{2n}}\sum\limits_{\lb\textbf{a},\textbf{b}\rb\in\lb \Z_d\times\Z_d\rb^n} W_{\lb\textbf{a},\textbf{b}\rb}^\dagger \n(W_{\lb\textbf{a},\textbf{b}\rb}X W_{\lb\textbf{a},\textbf{b}\rb}^\dagger)W_{\lb\textbf{a},\textbf{b}\rb}=\widetilde{\n}(X),
	\end{align}\label{equ:twirlingweyl}
	where
	\begin{align}
	\widetilde{\n}(X)=d^{-n}\sum\limits_{{\lb\textbf{a},\textbf{b}\rb}\in \lb \Z_d\times\Z_d\rb^n}\mu(\textbf{a},\textbf{b})\tr{W_{\lb\textbf{a},\textbf{b}\rb}^\dagger X}W_{\lb\textbf{a},\textbf{b}\rb} .  
	\end{align}
\end{prop}
\begin{proof}
	The map $\n\mapsto \widetilde{\n}$ defined in~\eqref{equ:twirlingweyl} is called the twirling of the quantum channel with respect to to the Weyl group~\cite{Siudzinska2018}. This is a linear projection map. Note that $\widetilde{\n}$ is a covariant quantum channel with respect to the Weyl group.
	As the (normalized) Weyl operators form an orthonormal basis for $\M_{d^n}$, we may expand any quantum channel $\n$ as a linear combination of maps of the form
	\begin{align}
	X\mapsto d^{-n}\tr{W_{\lb\textbf{a}_1,\textbf{b}_1\rb}^\dagger X} W_{\lb\textbf{a}_2,\textbf{b}_2\rb},
	\end{align}
	i.e. they also form a basis for the set of linear maps $\M_{d^n}\to\M_{d^n}$. By the linearity of the twirling operation, it suffices to analyse the effect of twirling on this basis.
	From~\eqref{equ:conjugationwithweyl} we see that the maps are invariant under twirling in the case $\lb\textbf{a}_1,\textbf{b}_1\rb=\lb\textbf{a}_2,\textbf{b}_2\rb$. Moreover, by Prop.~\ref{prop:characterizationweyl} that these also span the space of Weyl-covariant maps. Thus, it follows that the maps with $\lb\textbf{a}_1,\textbf{b}_1\rb\not=\lb\textbf{a}_2,\textbf{b}_2\rb$ are mapped to $0$, as twirling is a projection. This can also be easily seen by direct inspection. Expanding the quantum channel with respect to the Weyl basis concludes the proof.
\end{proof}
\subsection{Expectation values of the Weyl randomized benchmarking protocol}\label{extraclifford}
With the help of the following theorem, we can relate the measurement statistics of the protocol to the diagonals of the quantum channel in the Weyl basis.
\begin{thm}\label{thm:correctstatistics}
	Let $U\in\M_n$ be a unitary and $\n:\M_{d^n}\to\M_{d^n}$ be a quantum channel that encodes the noise after implementing $U$. Then, for a given $\lb\textbf{a},\textbf{b}\rb\in\lb\Z_d\times\Z_d\rb^n$ and sequence length $m$, initial state $\rho$ and POVM $E$ of the WRB protocol the output $X$ satisfies:
	\begin{align}
	\mathbb{E}(X)=d^{-n}\mu(\textbf{a},\textbf{b})^m\tr{W_{(\textbf{a},\textbf{b})}^\dagger \rho} \tr{E W_{\lb\textbf{a},\textbf{b}\rb}},
	\end{align}
	where 
	\begin{align}\label{equ:definitionmu}
	\mu\lb\textbf{a},\textbf{b}\rb=d^{-n}\tr{W_{(\textbf{a},\textbf{b})}^\dagger( \n\circ\mathcal{U})(W_{(\textbf{a},\textbf{b})})}.
	\end{align}
\end{thm}
\begin{proof}
	It follows from~\cite[Corollary 14]{Stilck_2018} that the expected channel at every step given by $\mathcal{S}=\widetilde{\n\circ\mathcal{U}}$, where this is the quantum channel produced by twirling the channel $\n\circ\mathcal{U}$ with respect to to the Weyl group. Using Prop.~\ref{prop:structuretwirling} the channel we implement at each step of the randomized benchmarking protocol is given by:
	\begin{align}
	\mathcal{S}(X)=d^{-n}\sum\limits_{{\lb\textbf{a},\textbf{b}\rb}\in \lb \Z_d\times\Z_d\rb^n}\mu(\textbf{a},\textbf{b})\tr{W_{\lb\textbf{a},\textbf{b}\rb}^\dagger X}W_{\lb\textbf{a},\textbf{b}\rb}.
	\end{align}
	This channel is diagonal, and the $m$-fold application of $S$ has the form
	\begin{align}
	\mathcal{S}^m(X)=d^{-n}\sum\limits_{{\lb\textbf{a},\textbf{b}\rb}\in \lb \Z_d\times\Z_d\rb^n}\mu(\textbf{a},\textbf{b})^m \tr{W_{\lb\textbf{a},\textbf{b}\rb}^\dagger X}W_{\lb\textbf{a},\textbf{b}\rb}.
	\end{align}
	If the initial random gate is $W_{\lb\textbf{a}_0,\textbf{b}_0\rb}$, and the probability of observing the outcome $E$ is given by
	\begin{align}
	p\lb\textbf{a}_0,\textbf{b}_0\rb=\tr{E \mathcal{S}^m(W_{\lb\textbf{a}_0,\textbf{b}_0\rb}\rho W_{\lb\textbf{a}_0,\textbf{b}_0\rb}^\dagger)}.
	\end{align}
	Thus, the expected value of $X$ is given by
	\begin{align}\label{equ:expectationvalue}
	&\mathbb{E}(X)=\nonumber\\
	&d^{-2n}\sum\limits_{\lb\textbf{a}_0,\textbf{b}_0\rb\in \lb \Z_d\times\Z_d\rb^n}\chi_{\lb-\textbf{b},\textbf{a}\rb}\lb\textbf{a}_0,\textbf{b}_0\rb p\lb\textbf{a}_0,\textbf{b}_0\rb=\nonumber\\
	&\tr{E \mathcal{S}^m(\mathcal{P}_{\lb\textbf{a},\textbf{b}\rb}(\rho))},
	\end{align}
	where we used Eq.~\eqref{equ:projectionirrep}. Now, from Eq.~\eqref{equ:projectionirrepweyl} we know that
	\begin{align}
	\mathcal{P}_{\lb\textbf{a},\textbf{b}\rb}(\rho)=d^{-n}\tr{W_{\lb\textbf{a},\textbf{b}\rb}^\dagger \rho} W_{\lb\textbf{a},\textbf{b}\rb},
	\end{align}
	Inserting this into~\eqref{equ:expectationvalue} proves the claim.
\end{proof}
Thus, using the character randomized benchmarking trick~\cite{1806.02048}, we may isolate each one of the diagonal elements and do the exponential fitting of one element at a time. This leads to increased numerical stability and allows for a clean analysis of the sample complexity of the protocol.

Finally, let us now show how assuming access to additional noiseless Clifford gates as a resource we can also access off-diagonal elements of the channel through a randomized benchmarking experiment:
\begin{cor}\label{cor:randimizedbenchclifford}
	Let $U\in\M_{d^n}$ be a unitary and $\n:\M_{d^n}\to\M_{d^n}$ be a quantum channel that describes the noise after implementing $U$. 
	Suppose that after implementing $U$ we implement a noiseless Clifford gate $C$ and that for a given $\lb\textbf{a}_1,\textbf{b}_1\rb\in\lb\Z_d\times\Z_d\rb^n$ we have
	\begin{align}\label{equ:actionclifford}
	C^\dagger W_{\lb\textbf{a}_1,\textbf{b}_1\rb}^\dagger C=e^{i\phi }W_{\lb\textbf{a}_2,\textbf{b}_2\rb}.
	\end{align}
	for some $\phi\in\R$ and $\lb\textbf{a}_2,\textbf{b}_2\rb\in\lb\Z_d\times\Z_d\rb^n$.
	Then, for a given sequence length $m$, initial state $\rho$ and POVM $E$ of the WRB protocol the output $X$ satisfies:
	\begin{equation}\label{equ:statistics_diagonal}
	\begin{split}
	&\mathbb{E}(X)=\\&d^{-n}\mu(\lb\textbf{a}_1,\textbf{b}_1\rb,\lb\textbf{a}_2,\textbf{b}_2\rb)^m\tr{W_{\lb\textbf{a}_1,\textbf{b}_1\rb}^\dagger \rho} \tr{E W_{\lb\textbf{a}_1,\textbf{b}_1\rb}},
	\end{split}
	\end{equation}
	where 
	\begin{equation}
	\begin{split}
	\mu(\lb\textbf{a}_1,\textbf{b}_1\rb,&\lb\textbf{a}_2,\textbf{b}_2\rb)=\\&d^{-n}e^{i\phi} \tr{W_{\lb\textbf{a}_2,\textbf{b}_2\rb}^{\dagger}\left[ \n\circ\mathcal{U}\right](W_{\lb\textbf{a}_1,\textbf{b}_1\rb})}.
	\end{split}
	\end{equation}
\end{cor}
\begin{proof}
	Using Theorem~\ref{thm:correctstatistics} that the expectation value of the protocol is given by:
	\begin{align}
	\mathbb{E}(X)=d^{-n}\mu(\textbf{a},\textbf{b})^m\tr{W_{\lb\textbf{a},\textbf{b}\rb}^\dagger \rho} \tr{E W_{\lb\textbf{a},\textbf{b}\rb}},
	\end{align}
	where 
	\begin{equation}
	\begin{split}
	\mu(\textbf{a},\textbf{b})&=d^{-n}\tr{W_{(\textbf{a},\textbf{b})}^\dagger\left[ \mathcal{C}\circ \n\circ\mathcal{U}\right](W_{(\textbf{a},\textbf{b})})}=\\
	&d^{-n}\tr{C^\dagger W_{(\textbf{a},\textbf{b})}^\dagger C\left[  \n\circ\mathcal{U}\right](W_{(\textbf{a},\textbf{b})})}.
	\end{split}
	\end{equation}
	Inserting Eq.~\eqref{equ:actionclifford} into the equation above yields the claim.
\end{proof}
Thus, having access to noiseless Clifford gates, it is possible to access any off-diagonal entry of the channel through Weyl randomized benchmarking experiments and, in principle, do complete tomography of the channel $\n\circ\mathcal{U}$.

The above results came with the caveat that  we can implement Weyl unitaries noiselessly. We will relax this assumption and show that the protocol still gives us valuable information as long as we assume that the application of each Weyl operator is followed by the same Weyl-diagonal channel $\mathcal{T}$.
\begin{lem}
	Let $U\in\M_{d^n}$ be a unitary and $\n:\M_{d^n}\to\M_{d^n}$ be a quantum channel that describes the noise after implementing $U$. Suppose that we run RWB where all unitary gates corresponding to Weyl operators are followed by the same Weyl diagonal channel $\mathcal{T}:\M_{d^n}\to\M_{d^n}$. That is, instead of implementing $\mathcal{W}_{(\textbf{a},\textbf{b})}$ we implement $\mathcal{T}\circ\mathcal{W}_{(\textbf{a},\textbf{b})}$. Then (using the same setting and notation as in Theorem~\ref{thm:correctstatistics}), we have:
	\begin{equation}
	\begin{split}
	\mathbb{E}(X)=d^{-n}\mu(\textbf{a},\textbf{b})^m&\mu_{W}(\textbf{a},\textbf{b})^{2m}\times\\ &\tr{W_{\lb\textbf{a},\textbf{b}\rb}^\dagger \rho} \tr{E W_{\lb\textbf{a},\textbf{b}\rb}}
	\end{split}
	\end{equation}
	with
	\begin{align}
	\mu_{W}(\textbf{a},\textbf{b})=d^{-n}\tr{W_{\textbf{a},\textbf{b}}^\dagger \mathcal{T}(W_{\textbf{a},\textbf{b}})}.
	\end{align}
\end{lem}
\begin{proof}
	Under the above assumptions, when we aim to implement the gate sequence $W_{\lb\textbf{a}_2,\textbf{b}_2\rb}\circ U,W_{\lb\textbf{a}_1\circ \textbf{b}_1\rb}$, we actually implement the noisy sequence of channels given by
	\begin{align}\label{equ:sequencenoisyweyls}
	\mathcal{T}\circ \mathcal{W}_{\lb\textbf{a}_2,\textbf{b}_2\rb}\circ \n \circ \mathcal{U}\circ \mathcal{T}\circ \mathcal{W}_{\lb\textbf{a}_1,\textbf{b}_1\rb}.
	\end{align}
	We know that Weyl diagonal channels are covariant with respect to the Weyl group:
	\begin{align}
	\mathcal{T}\circ\mathcal{W}_{(\textbf{a},\textbf{b})}=\mathcal{W}_{(\textbf{a},\textbf{b})}\circ \mathcal{T}.
	\end{align}
	Using this property in Eq.~\eqref{equ:sequencenoisyweyls}, we see that:
	\begin{equation}
	\begin{split}
	&\mathcal{T}\circ \mathcal{W}_{\lb\textbf{a}_2,\textbf{b}_2\rb}\circ \n\circ \mathcal{U}\circ \mathcal{T}\circ \mathcal{W}_{\lb\textbf{a}_1,\textbf{b}_1\rb}=\\&\mathcal{W}_{\lb\textbf{a}_2,\textbf{b}_2\rb}\circ \mathcal{T}\circ \n\circ \mathcal{U}\circ \mathcal{T}\circ \mathcal{W}_{\lb\textbf{a}_1,\textbf{b}_1\rb}.
	\end{split}
	\end{equation}
	Performing the protocol with the noisy Weyl operators gives rise to the same statistics as before, but with the channel $\mathcal{T}\circ \n\circ \mathcal{U}\circ \mathcal{T}$ instead of $\n\circ \mathcal{U}$ before. More precisely:
	\begin{equation}
	\begin{split}
	&\mathbb{E}(X)=d^{-n}\tr{W_{(\textbf{a},\textbf{b})}^\dagger \rho} \tr{E W_{\lb\textbf{a},\textbf{b}\rb}}\times \\
	&d^{-n}\tr{W_{(\textbf{a},\textbf{b})}^\dagger \left[\mathcal{T}\circ \n\circ \mathcal{U}\circ \mathcal{T}\right](W_{(\textbf{a},\textbf{b})})}^m.
	\end{split}
	\end{equation}
	Using the property that $\mathcal{T}$ is Weyl diagonal we get:
	\begin{align}
	&\tr{W_{(\textbf{a},\textbf{b})}^\dagger\left[\mathcal{T}\circ \n\circ \mathcal{U}\circ \mathcal{T}\right](W_{(\textbf{a},\textbf{b})})}=\\&\mu_{W}(\textbf{a},\textbf{b})^2\mu(\textbf{a},\textbf{b}).\nonumber
	\end{align}
\end{proof}
As long as the noise affecting Weyl operators is known and uniform, then we can also use the same protocol and extract the information from the diagonals. We note that the same conclusion also holds for the protocol with an extra Clifford gate: the expectation in the statement of Corollary~\ref{cor:randimizedbenchclifford} is replaced with
\begin{equation}
\begin{split}
\mathbb{E}(X)&=\mu_W(\textbf{a}_1,\textbf{b}_1)^{2m}\mu(\lb\textbf{a}_1,\textbf{b}_1\rb,\lb\textbf{a}_2,\textbf{b}_2\rb)^m\times \\&d^{-n}\tr{W_{\lb\textbf{a}_1,\textbf{b}_1\rb}^\dagger \rho} \tr{E W_{\lb\textbf{a}_1,\textbf{b}_1\rb}}. 
\end{split}
\end{equation}

\subsection{Example choices of initial state and POVM}\label{sec:initial_choices}
Let us now discuss how to pick the initial states $\rho$ and the POVM $E$ for the randomized benchmarking protocol. We want to fit the expression in Eq.~\eqref{equ:statistics_diagonal} to an exponential curve, and thus it may be advantageous to ensure that the term $d^{-n}\tr{W_{\lb\textbf{a}_1,\textbf{b}_1\rb}^\dagger \rho} \tr{E W_{\lb\textbf{a}_1,\textbf{b}_1\rb}}$ is of constant order. Indeed, if this this term is too small, then estimating this expectation value might require a prohibitive number of samples. Therefore, the constant order approximation is the best one could hope for. Indeed, it follows from a H\"older inequality:
\begin{align*}
    &|\tr{W_{\lb\textbf{a}_1,\textbf{b}_1\rb}^\dagger \rho}|\leq \| W_{\lb\textbf{a}_1,\textbf{b}_1\rb}\|_{\infty}\|\rho\|_1=1\\
    &|\tr{E W_{\lb\textbf{a}_1,\textbf{b}_1\rb}}|\leq\|E\|_\infty\|W_{\lb\textbf{a}_1,\textbf{b}_1\rb}\|_{1}\leq  d^n.
\end{align*}
A canonical choice for $\rho$ is an eigenstate of $W_{\lb\textbf{a}_1,\textbf{b}_1\rb}$. As $W_{\lb\textbf{a}_1,\textbf{b}_1\rb}$ is a product observable, this can be chosen as a product state. Therefore,
\begin{align*}
\left|\tr{W_{\lb\textbf{a}_1,\textbf{b}_1\rb}^\dagger \rho}\right|=1.
\end{align*}
A canonical choice for $E$ is the projector onto an eigenspace  of dimension $d^{n-1}$ of $W_{\lb\textbf{a}_1,\textbf{b}_1\rb}$. To see that such an eigenspace exists, note that as $(W_{\lb\textbf{a}_1,\textbf{b}_1\rb})^d=I$, all the eigenvalues of $W_{\lb\textbf{a}_1,\textbf{b}_1\rb}$ are  $d$-th roots of unity. As there are only $d$ possible eigenvalues for $W_{\lb\textbf{a}_1,\textbf{b}_1\rb}$ not counting multiplicities, there must be at least one eigenspace of dimension $d^{n-1}$. Let $E$ be the projector onto that  eigenspace. Then:
\begin{align*}
   \left|\tr{E W_{\lb\textbf{a}_1,\textbf{b}_1\rb}}\right|\geq d^{n-1}. 
\end{align*}
To measure $E$ we can measure in the eigenbasis of $W_{\lb\textbf{a}_1,\textbf{b}_1\rb}$.
As $W_{\lb\textbf{a}_1,\textbf{b}_1\rb}$ is a tensor product observable, this can be achieved by implementing a quantum circuit of depth $1$. Thus, with this choice we achieve
\begin{align*}
   \left| d^{-n}\tr{W_{\lb\textbf{a}_1,\textbf{b}_1\rb}^\dagger \rho} \tr{E W_{\lb\textbf{a}_1,\textbf{b}_1\rb}}\right|\geq\frac{1}{d}.
\end{align*}
As discussed before, the exact value of $d^{-n}\tr{W_{\lb\textbf{a}_1,\textbf{b}_1\rb}^\dagger \rho} \tr{E W_{\lb\textbf{a}_1,\textbf{b}_1\rb}}$ is not of essence to the protocol, as long as it is not too small. The above choice achieves that by only requiring product measurement and initial states, making it a good pick.

\subsection{Examples: Dephasing and depolarizing noise}\label{app:depdepnoise}

To illustrate how our protocol handles different types of noise, we look at four important cases of local depolarizing or dephasing noise with parameters $p_1$ and $p_2$ and global dephasing or depolarising noise with parameter $p$ acting on $2$ qudits. The action of the above channels on the elements of Weyl basis are presented in Table~\ref{action_WeylChannels} below.
\begin{table}[h]\label{action_WeylChannels} 
	\begin{tabular}{ |c|c|c|c|c|} 
		\hline
		Channels & $W_{0,0}W_{0,1}$ & $W_{0,1}W_{1,1}$ & $W_{1,1}W_{0,1}$ & $W_{1,1}W_{1,1}$ \\ 
		\hline
		Local dephasing & $1$ & $p_2$ & $p_1$ & $p_1p_2$ \\ 
		\hline
		Local dephasing & $1$ & $p$ & $p$ & $p$ \\ 
		\hline
		Local depolarising & $p_2$ & $p_1p_2$ & $p_1p_2$ & $p_1p_2$\\  
		\hline
		Global depolarising & $p$ & $p$ & $p$ & $p$\\ 
		\hline
	\end{tabular}
	\caption{The action of local depolarizing, dephasing noise with parameters $p_1,p_2$, together with global dephasing and depolarizing noise with parameter $p$ acting on two qudits.}
\end{table}

\section{Making Weyl randomized benchmarking efficient}\label{sec:effbench}
Weyl randomized benchmarking allows to identify a variety of experimentally relevant noise models that affect each layer of a general unitary circuit.
However, the number of parameters to be estimated in a Weyl randomized benchmarking experiment on $n$-qudits is $d^{2n}$ which is not feasible even for a moderate number of qudits. Thus, it becomes necessary to make further restrictions on the noise models to render this protocol efficient. In this section we will discuss how assumptions on the locality of the noise can be used to achieve this goal and render the protocol practical. 
\subsection{Local noise models}
We now turn to restricted noise models by imposing a certain locality structure. To model this this situation, we start from a physically motivated hypergraph $G=(V,E)$. The hyperedges in $E$ encode the interactions between the subsystems and by extension -- the locality of the noise. Given a hyperedge $e\in E$, we denote $f_e:(\Z_d\times\Z_d)^n\to\C$ to be a function such that $f_e(\textbf{a},\textbf{b})$ only depends on the value of $(\textbf{a},\textbf{b})$ on the substring on $e$. We then have:

\begin{defn}[Local Weyl channel]
	Given a hypergraph $G=(V,E)$ with $|V|=n$, we call a Weyl diagonal quantum channel $\mathcal{T}:\M_{d^n}\to\M_{d^n}$ physically local with respect to to $G$ if we can express the action of $\mathcal{T}$ on any $W_{(\textbf{a},\textbf{b})}$ as:
	\begin{align}
	\mathcal{T}(W_{(\textbf{a},\textbf{b})})=\sum\limits_{e\in E}f_e(\textbf{a},\textbf{b})W_{(\textbf{a},\textbf{b})}.
	\end{align}
\end{defn}
To illustrate this definition, consider the following example: 
\begin{example}
	Suppose that our system consists of $4$ qudits and our hypergraph $G$ is a circle. 
	Denote by $\mathcal{R}_{(i,j)}$ a channel that conjugates qudits $i,j$ with $W_{(1,0)}\otimes W_{(1,0)}$ and acts as the identity on the rest.
	An example of a local Weyl diagonal channel with respect to this graph is then:
	\begin{align}
	\mathcal{T}=\frac{1}{4}\left(\mathcal{R}_{(1,2)}+\mathcal{R}_{(2,3)}+\mathcal{R}_{(3,4)}+\mathcal{R}_{(4,1)}\right).
	\end{align}
	More generally, a convex combination of unitary channels consisting of Weyl conjugations only acting on qubits connected by an edge give rise to physically local Weyl-channels.
\end{example}

\begin{example}
	Suppose that we have a system consisting of $4$ qubits. 
	Denote $\mathcal{R}_{(i,j)}$ as in Example 1, but pick the hypergraph to be a complete graph on $4$ vertices.
	This induces the following diagonal channel:
	\begin{align}
	\mathcal{T}=\frac{1}{6}\left(\mathcal{T}_{(1,2)}+\mathcal{T}_{(2,3)}+\mathcal{T}_{(3,4)}+\mathcal{T}_{(4,1)}+\mathcal{T}_{(2,4)}+\mathcal{T}_{(1,3)}\right).
	\end{align}
	Note that this channel is not local with respect to the circle hypergraph, as its spectrum depends on Weyl operators in, say, $(1,3)$, which were not an edge in the previous hypergraph.
\end{example}

\begin{example}\label{example:klocalnoise}
	One could also be agnostic with respect to locality of the errors and assume that the noise can act on at most $k<n$ qudits at a time. The underlying hypergraph would then be the complete hypergraph with hyperedges of size $k$.
	The number of parameters necessary to describe such a quantum channel then scales like ${n \choose k} d^{2k}$.
\end{example}

Example~\ref{example:klocalnoise} reflects the scenario where we assume that, up to small corrections, the system is affected by errors acting on at most $k$ out of the $n$ qubits, for $k$ some constant. Thus, our protocol allows for the characterization of such noise channels in polynomial  time.

Indeed, imposing such natural restrictions on noise models significantly reduces the number of parameters one needs to fit. This is because the noise model is completely determined by $f_e$. As each of these functions depends only on Weyl operators in the subsystems included in $e$, and each of these functions has $d^{2|e|}$ parameters the total number of parameters to fit is
\begin{align}\label{equ:numberofparameters}
\sum_{e\in E}d^{2|e|}.
\end{align}

We will now discuss how to extract the functions $f_e$ from the randomized benchmarking experiment, as they completely characterize the channels that are local with respect to a hypergraph.
\subsection{Fitting of the parameters for local noise}\label{sec:nocross}
The first step is to relate the noise parameters to the results of the randomized benchmarking. 

\begin{prop}[Expectation values of local Weyl diagonal channels]
	Let $\mathcal{T}$ be a local Weyl diagonal channel with respect to a hypergraph $G=(V,E)$. Then for a Weyl operator $W_{(\textbf{a},\textbf{b})}$ we have:
	\begin{equation}
	\begin{split}
	&\tr{W_{(\textbf{a},\textbf{b})}^{\dagger}(\mathcal{T}\circ\mathcal{U})(W_{(\textbf{a},\textbf{b})})}=\\&\tr{W_{(\textbf{a},\textbf{b})}^{\dagger}\mathcal{U}(W_{(\textbf{a},\textbf{b})})}\times\lb \sum\limits_{e\in E}f_e(\textbf{a},\textbf{b})\rb.
	\end{split}
	\end{equation}
\end{prop}
\begin{proof}
	This follows from expanding $\mathcal{T}(W_{(\textbf{a},\textbf{b})})$ in the Weyl basis, and using the action of $\mathcal{T}$ on each element of the Weyl basis
	\begin{align}
	\mathcal{T}(W_{(\textbf{a},\textbf{b})})=\sum\limits_{e\in E}f_e(\textbf{a},\textbf{b})W_{(\textbf{a},\textbf{b})}.
	\end{align}
\end{proof}
The last proposition enables us to relate the results of the randomized benchmarking experiments to the eigenvalues of the Weyl diagonal quantum channels. The randomized benchmarking experiment gives us access to $\tr{W_{(\textbf{a},\textbf{b})}^{\dagger}(\mathcal{T}\circ\mathcal{U})(W_{(\textbf{a},\textbf{b})})}$.
Thus, using the last proposition together with the results of the randomized benchmarking experiment with knowledge of $\tr{W_{(\textbf{a},\textbf{b})}^{\dagger}\mathcal{U}(W_{(\textbf{a},\textbf{b})})}$ we get a linear equation for the eigenvalues of the local noise channels.

Moreover, note that as $\mathcal{T}$ is a quantum channel, we get
\begin{align}
\sum\limits_{e}f_e(\textbf{0},\textbf{0})=1.
\end{align}
Given a hypergraph, the number of parameters we need to fit is given in Eq.~\eqref{equ:numberofparameters}. This is also a number  of linearly independent randomized benchmarking results required to completely determine the parameters of the functions $f_e$ by solving a linear system of equations or, if more data is available, performing linear regression.

We will discuss the stability of this procedure and the necessary number of samples in the next section.

\section{Statistical and stability analysis of the RB protocol}\label{app:statstability}
The results of the last section raise some technical questions:
\begin{enumerate}
    \item what is the number of samples required to get to a given level of confidence about the range of the parameters in a given noise model.
    \item how to pick the parameters of the randomized benchmarking protocol and how robust it is. 
\end{enumerate}
The goal of this section is to answer these questions in a rigorous way. 
\subsection{Length of the test sequences}\label{sequencelength}
We will now discuss how to pick the sequence lengths to ensure reliable results.
We present the generalized approach of~\cite{Harper2019} for this randomized benchmarking procedure, with the main technical difference being that the decay rate $\mu(a,b)$ may be complex. As discussed before, the goal of the randomized benchmarking protocol is to estimate the diagonal elements of the evolution with respect to the Weyl basis, as these can be related to the noise parameters.
We will first discuss how to pick the length of the sequence of gates to get good estimates. As remarked in~\cite{Harper2019}, it is highly desirable to get multiplicative bounds instead of additive ones on the diagonals $\mu(\textbf{a},\textbf{b})$. To illustrate why this is important, consider an example of a system suffering from global depolarizing at rate $p$. Moreover, note that the number of samples required to estimate the expectation value of random variable up to an additive error $\epsilon$ usually scales like $\epsilon^{-2}$.
Current state of the art implementations have $p\simeq 1-10^{-3}$, which implies that order $10^{6}$ runs of the experiment would be necessary to get an additive error of the same order as the actual parameter, which is too costly. Thus, it is desirable to have bounds which are multiplicative, i.e., scale in $1-p$ instead of just having an additive error. 
To see how this issue relates to the length of the sequence of gates and multiplicative bounds, note that if $p$ is of order $1-10^{-3}$, then after a sequence of length $100$, the probability of success for the randomized benchmarking experiment will have reduced to roughly $1-10^{-2}$, while performing the same experiment for $p$ of order $1-10^{-2}$ will have reduced the success probability to roughly $0.37$. We see that in this case, we are able to tell these two scenarios apart with an additive error of order $10^{-1}$ on the estimates.
On the other hand, if the gate sequence is too long, then in both cases the success probability will be too small and it will not be possible to tell them apart reliably.
This indicates that the sequence length should be chosen in a way that ensures that the survival probability is of constant order $\gg 0$, so that a (not too small) additive error is enough to ensure that the estimate is reliable. Let us formalize this intuition.

As mentioned before, our setting presents some additional challenges when compared to that of~\cite{Harper2019} because the diagonal elements we want to estimate can also be complex. Thus, we will discuss how to estimate the absolute value of the diagonal elements and the corresponding phase separately.

As seen in Theorem~\ref{thm:correctstatistics}, if we perform the randomized benchmarking experiment with an initial state $\rho$ and measure the POVM $E$ for the group element $(\textbf{a},\textbf{b})\in\left(\mathbb{Z}_d\times\mathbb{Z}_d\right)^n$, then the weighted survival probability  at sequence length $m$ will have the expectation value:
\begin{align}
q(\textbf{a},\textbf{b},m)=C(\textbf{a},\textbf{b})\mu(\textbf{a},\textbf{b})^m,
\end{align}
where $\mu$ is defined as in Eq.~\eqref{equ:definitionmu} and
\begin{align}
C(\textbf{a},\textbf{b})=d^{-n}\tr{W_{\lb\textbf{a},\textbf{b}\rb}^\dagger \rho} \tr{E W_{\lb\textbf{a},\textbf{b}\rb}}.
\end{align}
As noted before, $\mu(\textbf{a},\textbf{b})$ is in general a complex number.
On the other hand, $C(\textbf{a},\textbf{b})$ is a real number. This is because the map $\mathcal{P}_{(\textbf{a},\textbf{b})}$ is hermiticity preserving, as noted in~\cite{Stilck_2018} and
\begin{align}
C(\textbf{a},\textbf{b})=\tr{E\mathcal{P}_{(\textbf{a},\textbf{b})}(\rho)}.
\end{align}
We will assume that $|C_{(\textbf{a},\textbf{b})}|\gg0$. In Sec.~\ref{sec:initial_choices} we discussed how to achieve this.
We will now drop the $(\textbf{a},\textbf{b})$ subscripts and arguments, as we will assume them to be fixed throughout the rest of this subsection, that is, we will be interested in learning one parameter. 

As $\mu$ is a complex number, one way of specifying it is by estimating its phase $\phi$ and absolute value $|\mu|$ such that $\mu=|\mu|e^{i\phi}$.
We will focus on estimating $|\mu|$ with a multiplicative error first through the randomized benchmarking experiments. 
The first step towards estimating $\mu$ will be to specicfy how to estimate $|q(m)|^2$.
The following lemma provides a bound on the number of samples required to get an additive error estimate of $|q(m)|^2$ for a given sequence length. 
\begin{lem}\label{lem:hoeffdingsamples}
	Let $m$ be fixed and $\epsilon,\delta>0$ be given. Suppose we repeat the character randomized benchmarking experiment for $2l$ random sequences of gates, each of length $m$. Let $s_{k}(m)$ be the observed outcome for sequence $k$ and define the random variable $X_k$ for $1\leq k\leq l$ as:
	\begin{align}
	X_k=\operatorname{Re}\left(s_{k}(m) \bar{s}_{k+l}(m)\right).
	\end{align}
	Then, for $l=\cO(\epsilon^{-2}\log\lb \delta^{-1}\rb)$ with probability at least $1-\delta$ we get:
	\begin{align}
	\left|\frac{1}{l}\sum\limits_{k=1}^{l}X_k-|q(m)|^2\right|\leq \epsilon.
	\end{align}
\end{lem}
\begin{proof}
	From the above description
	\begin{align}
	\mathbb{E}[ s_{k}(m)]=   q(m).
	\end{align}
	Thus,  as $s_{k}$ and $ s_{k+l}]$ are independent random variables and by the linearity of expectation values, we have:
	\begin{align}
	\mathbb{E}[ X_k]=\mathbb{E}[ s_{k}]\mathbb{E}[ \bar{s}_{k+l}]=|q(m)|^2
	\end{align}
	and
	\begin{align}
	\mathbb{E}[ \text{Re}X_k]=\text{Re}\mathbb{E}[ X_k].
	\end{align}
	Also note that $|X_k|\leq1$, as the output of the randomized benchmarking protocol is either some complex number of modulus $1$ or $0$ and, thus, $|s_{k}(m) s_{k+l}(m)|\leq1$.
	The claim then follows from Hoeffding's inequality.
\end{proof}
We have shown that the empirical average of the $X_k$ provides an estimator for $|q(m)|^2$ up to an additive precision. Leveraging on that let us now show how to estimate the absolute value. The procedure on how to pick the sequence length is displayed in ~\ref{fig:estimatormu1}.
We now adapt the results of~\cite{Harper2019} to show that the output of the procedure in Fig.~\ref{fig:estimatormu1} satisfies:

	\begin{tcolorbox}[breakable]\label{fig:estimatormu1}
		\emph{Input:} desired precision $\epsilon>0$, upper bound $u$ on $|\mu\|^2$.\\ 
\emph{Output:} Estimate $\hat{\mu}$ of $\mu$.
		\begin{enumerate}
		\itemsep-0.2em  
			\item Set $m_1=1$. Produce an estimate $|\hat{q}(1)|^2$ of $|q(1)|^2$ up to additive error $u^2C^{2}\epsilon$.
			\item While $|\hat{q}(m_i)|^2>|\hat{q}(1)|^2/3$:
			\begin{itemize}
				\item Set $m_i:=2^i+1$.
				\item Estimate $|q(m_i)|^2$ up to an additive error $|\mu|^2C^2\epsilon$ and set it to $|\hat{q}(m_i)|^2$.
			\end{itemize}
			\item Output $|\hat{\mu}|=\lb \frac{\hat{q}(m_i)}{\hat{q}(m_1)}\rb^{\frac{1}{2m_i}}$.
		\end{enumerate}

	\end{tcolorbox}
	
\begin{prop}[Multiplicative estimates for absolute value of diagonal]\label{prop:qualityest}
	The estimate $|\hat{\mu}|$ outputted by the algorithm in Fig.~\ref{fig:estimatormu1} satisfies:
	\begin{align}\label{equ:qualityestimate}
	|\mu|-|\hat{\mu}|=\cO(\epsilon(1-|\mu|)).
	\end{align}
\end{prop}
\begin{proof}
	Let $m$ be the sequence length of the output.
	From the assumptions on the error of the estimates of $q(m_i)$ we have:
	\begin{equation}
	\begin{split}
	&\lb \frac{C^2|\mu|^{2m+2}-\epsilon |\mu|^2C^2}{C^2|\mu|^2+|\mu|^2\epsilon C^2}\rb\leq \lb \frac{\hat{q}(m_i)}{\hat{q}(m_1)}\rb^{\frac{1}{2m}},\\
	&\lb \frac{\hat{q}(m_i)}{\hat{q}(m_1)}\rb^{\frac{1}{2m}}\leq  \lb \frac{C^2|\mu|^{2m+2}+\epsilon |\mu|^2C^2}{C^2|\mu|^2-\epsilon|\mu|^2 C^2}\rb.
	\end{split}
	\end{equation}
	Simplifying the expressions above we see that:
	\begin{align}\label{equ:upperboundestimator}
	\lb \frac{\hat{q}(m_i)}{\hat{q}(m_1)}\rb^{\frac{1}{2m}}\leq |\mu|\lb \frac{1+\epsilon|\mu|^{-2m}}{1-\epsilon}\rb^{\frac{1}{2m}}.
	\end{align}
	When the condition in the while loop is true, we have
	\begin{align}
	C^2|\mu|^{m+1}+C^2\epsilon|\mu|^2\geq\frac{1}{3}\lb C^2|\mu|^2+C^2\epsilon|\mu|^2\rb,
	\end{align}
	as the previous step to $m$ had sequence length $(m-1)/2+1$ by construction. Simplifying and squaring the inequality we get:
	\begin{align}
	|\mu|^{2m+1}\geq \frac{1}{9}\lb 1-17\epsilon\rb.
	\end{align}
	Note that for quantum channels we have $|\mu|\leq1$ and, thus:
	\begin{align}
	|\mu|^{2m}\geq\frac{1}{9}\lb 1-17\epsilon\rb.
	\end{align}
	Inserting the inequality above into Eq.~\eqref{equ:upperboundestimator} we get:
	\begin{align}\label{equ:upperboundoutput2}
	&\lb \frac{\hat{q}(m_i)}{\hat{q}(m_1)}\rb^{\frac{1}{2m}}\leq |\mu|\lb \frac{1+9\epsilon\lb 1-17\epsilon\rb^{-1}}{1-\epsilon}\rb^{\frac{1}{2m}}=\nonumber\\
	&|\mu|\lb 1+\cO(\epsilon)\rb^{\frac{1}{2m}}.
	\end{align}
	Again, by our stopping criterion: 
	\begin{align}
	m=\Theta(\log(|\mu|^{-1})^{-1}),
	\end{align}
	which gives that $m^{-1}=\cO(1-|\mu|)$, as $\log(x)\geq1-1/x$ for $x>0$.

	Thus, we see that
	\begin{equation}
	\begin{split}
	&\lb 1+\cO(\epsilon)\rb^{\frac{1}{2m}}=\text{exp}\lb\frac{\log(1+\cO(\epsilon))}{2m}\rb\leq \\&\text{exp}\lb\cO(\epsilon(1-|\mu|)\rb=\lb1+\cO[\epsilon(1-|\mu|)]\rb.
	\end{split}
	\end{equation}
	
Combining the bound above with Eq.~\eqref{equ:upperboundoutput2} we get that
	\begin{align}
	|\hat{\mu}|-|\mu|=\cO(\epsilon(1-|\mu|)).
	\end{align}
	The bound in the other direction follows analogously.
\end{proof}
Thus, as long as the estimator terminates, we get an estimate with multiplicative error of $|\mu|$.
It remains to compute the sequence length and the required number of samples after the procedure terminates. This is the content of the next theorem. To simplify our derivations, we introduce the concept of the Weyl spectral gap -- the largest diagonal entry with respect to the Weyl basis (excluding the identity):
\begin{defn}[Weyl spectral gap]
	Let $\n:\M_{d^n}\to\M_{d^n}$ be a quantum channel. Its Weyl spectral gap $\lambda$ is given by
	\begin{align}
	1-\lambda=\max\limits_{(\textbf{a},\textbf{b})\in\left(\mathbb{Z}_d\times\mathbb{Z}_d\right)^n),(\textbf{a},\textbf{b})\not=0} d^{-n}|\tr{ W_{(\textbf{a},\textbf{b})}^\dagger \n \lb W_{(\textbf{a},\textbf{b})}\rb}|.
	\end{align}
\end{defn}

\begin{thm}\label{thm:recoveryguarantees}
	Let $\n$ be a quantum channel with Weyl spectral gap $\lambda$, $\epsilon>0$ be a given error parameter satisfying $\epsilon\leq 200^{-1}C^2|\mu|^2$ and $1-\delta>0$ be a failure probability. Then the procedure above outputs an estimate $|\hat{\mu}|$ of $\mu$ satisfying Eq.~\eqref{equ:qualityestimate} with probability at least $1-\delta$ using
	\begin{align}
	\cO(\epsilon^{-2}(C\mu)^{-4}\log[(1-\lambda)^{-1}]\log(\delta^{-1}\log[(1-\lambda)^{-1}]))
	\end{align}
	samples and largest sequence length $M_{\max}$
	\begin{align}
	M_{\max}=\cO(\lambda^{-1}).
	\end{align}
\end{thm}
\begin{proof}
	The fact that the output satisfies Eq.~\eqref{equ:qualityestimate} follows from Lemma~\ref{prop:qualityest} after procedure terminates. Thus, it only remains to show how many samples are required to ensure that all estimates are correct until the algorithm terminates and the number of steps after which it terminates.
	First, let us estimate the expected sequence length which results in the termination. The termination criterion is $|\hat{q}(m_i)|^2\leq\frac{1}{3}|\hat{q}(m_1)|^2$.
	As in the previous lemma, as long as all the estimates are correct up to an additive error $\epsilon \lb C\mu\rb^2$ the procedure terminates whenever
	\begin{align}
	C^2|\mu|^{2m}+C^2\epsilon|\mu|^2\leq\frac{1}{3}\lb C^2|\mu|^2+C^2\epsilon|\mu|^2\rb.
	\end{align} 
	We assumed that $\epsilon$ satisfies $\epsilon\leq \frac{1}{200} |\mu|^2C^2$, and from the above equation we get
	\begin{align}
	m=\cO(\log(\mu)^{-1}).
	\end{align}
	By our assumption on the gap $|\mu|^{2m_i-1}\leq (1-\lambda)^{2m_i-1}$, thus picking
	\begin{align}
	M_{\max}=\cO\lb\log(\mu^{-1})\rb
	\end{align}
	is enough to ensure that the termination condition is satisfied. By the definition of the Weyl spectral gap we have $\mu\leq1-\lambda$, which yields the estimate on the largest sequence length.
	It now remains to compute the number of samples required to ensure that all estimates have the required precision with the desired failure probability.
	Recall that we set $m_{i}=2^i+1$. Thus, it follows from our estimate on $m_{\max}$ that the total number of iterations required by the algorithm is $i_{\max}=\cO\lb\log\lb\log((1-\lambda)^{-1})\rb\rb$. Thus, we need to estimate $|q(m_i)|^2$ correctly, i.e. up to an additive error of $\epsilon (C^2|\mu|^2)$, for $i_{\max}$ many different sequence lengths. 
	It follows from Lemma~\ref{lem:hoeffdingsamples} that $\cO(\epsilon^{-2}\lb C|\mu|\rb^{-4}\log(\delta^{-1} i_{\max}))$ many samples for each iteration suffice to ensure an additive error $\epsilon C^2|\mu|^2$ and failure probability at most $\delta i_{\max}^{-1}$ for every iteration. By the union bound, we see that the probability that all estimates are correct up to an additive error $\epsilon C^2|\mu|^2$ is at least $1-\delta$. Thus, we conclude that a total of 
	\begin{equation}
	\begin{split}
	&\cO(i_{\max}\lb C|\mu|\rb^{-4}\epsilon^{-2}\log(\delta^{-1} i_{\max}))=\\ &\cO(\epsilon^{-2}\lb C|\mu|\rb^{-4}\log[(1-\lambda)^{-1}]\log(\delta^{-1}\log[(1-\lambda)^{-1}]))
	\end{split}
	\end{equation}
	samples suffice to reach the desired accuracy.
\end{proof}
This theorem establishes the number of samples required to get a multiplicative estimate on $|\mu|$ and the maximum sequence length.
First, note that in the setting of our protocol $\mathcal{N}\circ\mathcal{U}$, when $U$ is the target unitary and $\mu$ the desired diagonal of this channel with respect  to the Weyl basis. Note that even in the case when $\mathcal{N}$ is the identity, it can be the case that $\mu$ is very small or zero. This is the case if the corresponding diagonal element of the unitary is small. Thus, the procedure is only effective if the diagonal element of the noiseless unitary also has constant order.  Otherwise we can use the noiseless Clifford trick discussed in Cor.~\ref{cor:randimizedbenchclifford} to access off-diagonal elements. This also  has to be taken into account when choosing the  desired precision. More precisely, suppose as usual that the channel $\mathcal{T}$ is Weyl diagonal with corresponding eigenvalues $\lambda(\textbf{a},\textbf{b})$. Then we have:
\begin{align}
|\mu(\textbf{a},\textbf{b})|=d^{-n}|\lambda(\textbf{a},\textbf{b})\tr{W_{(\textbf{a},\textbf{b})}^{\dagger}\mathcal{U}(W_{(\textbf{a},\textbf{b})})}|.
\end{align}
Thus, in order to estimate $|\lambda(\textbf{a},\textbf{b})|$, the desired parameter, we need to know the diagonal of the unitary. The error gets rescaled by $d^{-n}|\tr{W_{(\textbf{a},\textbf{b})}^{\dagger}\mathcal{U}(W_{(\textbf{a},\textbf{b})})}|^{-1}$, after we multiply our estimate on $|\mu|$ to estimate $|\lambda(\textbf{a},\textbf{b})|$. We conclude that:
\begin{cor}
	Let $\mathcal{T}\circ\mathcal{U}$ be a quantum channel for a known unitary $U$ and $\mathcal{T}$ a Weyl diagonal channel.
	Denote by $u(\textbf{a},\textbf{b})$ the diagonal elements of $\mathcal{U}$, i.e.
	\begin{align}
	u(\textbf{a},\textbf{b})=d^{-n}\tr{W_{(\textbf{a},\textbf{b})}^{\dagger}\mathcal{U}(W_{(\textbf{a},\textbf{b})})}
	\end{align}
	and similarly by $\lambda(\textbf{a},\textbf{b})$ those of $\mathcal{T}$. 
	Consider the setting of Theorem~\ref{thm:recoveryguarantees} with a given error parameter $\epsilon'$; pick $\epsilon=\epsilon' |u(\textbf{a},\textbf{b})^{-1}|$, where we further assume $|u(\textbf{a},\textbf{b})^{-1}|>0$.
	Then $|\hat{\lambda}(\textbf{a},\textbf{b}))|=|\mu u(\textbf{a},\textbf{b})^{-1}|$ satisfies:
	\begin{align}
	||\lambda(\textbf{a},\textbf{b}))|-|\hat{\lambda}(\textbf{a},\textbf{b}))||\leq \cO(\epsilon(1-|\mu|).
	\end{align}
\end{cor}
\begin{proof}
	The claim follows from the discussion above combined with Theorem~\ref{thm:recoveryguarantees} and Proposition~\ref{prop:qualityest}.
\end{proof}
Thus, as long as we can compute the diagonal elements of the unitary and they are not too small, we are able to recover the absolute value of the diagonals of the corresponding Weyl channel from them.
Note that this is all the information required to estimate the $\|\cdot\|_{\ell_1}$ norms underlying the complexity of the negativity algorithm.

We now show how one can to learn the phase of the corresponding diagonal elements and thus to completely characterize the noise.

\begin{prop}
	Consider fixed $\epsilon,\delta>0$, and $\theta\in(-\frac{\pi}{2},-\frac{\pi}{2})$. Suppose we repeat the character randomized benchmarking experiment for $2l$ random sequences of gates of length $m$. Let $s_{k}(m)$ be the observed outcome for sequence $k$ and define the random variables $Y_m,Z_m$:
	\begin{align}
	Y_m=\frac{1}{l}\sum\limits_{k=1}^l\operatorname{Re}\left(s_{k}(m)\right),Z_m=\frac{1}{l}\sum\limits_{k=1}^l\operatorname{Im}\left(s_{k}(m)\right).
	\end{align}
	Then, for $l=\cO(\epsilon^{-2}\log\lb \delta^{-1}\rb)$ with probability at least $1-\delta$ we get:
	\begin{align}
	\left|\arctan\lb \frac{Z_m}{Y_m} \rb-m\theta\right|\leq \lb \cO(|\mu| \rb^{-m}|C|\epsilon),
	\end{align}
	where $m\theta$ is taken modulo $2\pi$.
\end{prop}
\begin{proof}
	Note that 
	\begin{align}
	\mathbb{E}(\textrm{Im}\left(s_{k}(m)\right))=|\mu|^mC\sin(m\theta)
	\end{align}
	and
	\begin{align}
	\mathbb{E}(\textrm{Re}\left(s_{k}(m)\right))=|\mu|^mC\cos(m\theta).
	\end{align}
	Using Hoeffding's inequality, $l$ many samples suffice to ensure that:
	\begin{align}
	&\left|Y_m-|\mu|^mC\cos(m\theta)\right|\leq \epsilon,\\&\left|Z_m-|\mu|^mC\sin(m\theta)\right|\leq \epsilon.
	\end{align}
	Let $R_m=|\mu|^mC\cos(m\theta)$ and $I_m=|\mu|^mC\sin(m\theta)$.
	Using a Taylor expansion we see that:
	\begin{equation}
	\begin{split}
	&\arctan\lb \frac{R_m+\delta_1}{I_m+\delta_2}\rb= \arctan\lb \frac{R_m}{I_m}\rb-\frac{R_m}{R_m^2+I_m^2}\delta_2\\
	&-\frac{I_m}{R_m^2+I_m^2}\delta_1+\cO(\delta_1^2+\delta_2^2+\delta_1\delta_2).
	\end{split}
	\end{equation}
	
	Thus, if we have an estimate of $R_m$ and $I_m$ up to an error $\epsilon$, we get
	\begin{align}
	\left|m\theta-\arctan\lb \frac{Z_m}{Y_m}\rb\right|\leq \cO(\epsilon |\mu|^{-m}|C|^{-1}).
	\end{align}
\end{proof}
The last proposition tells us how to obtain an additive approximation of $m\theta$. 
It is then possible to get additive approximations for several different values of $m$ and perform a linear fitting to further improve the accuracy of the estimate.
The assumption that $\theta\in(-\frac{\pi}{2},\frac{\pi}{2})$ might seem restrictive at  first, but note that if $\theta$ does not lie in this interval, we can add $\frac{\pi}{2}$ to the estimate we  obtained.   Checking which is the case can be done by looking at the sign of $Y_m$; if it is positive, we make a guess that $\theta\in(-\pi/2,\pi/2)$, otherwise $\theta\in(\pi/2,3\pi/2)$. 
We will make the right guess as long as $|\cos(\theta)\mu C|\geq \epsilon$, i.e. the spectrum is not close to being strictly imaginary.
However, it is natural to make the following assumption on the underlying Weyl-diagonal channel:

\begin{defn}[Symmetric Diagonal Weyl-Channel]
	A Weyl-diagonal quantum channel $\mathcal{T}:\M_{d^n}\to\M_{d^n}$ is symmetric if $\mathcal{T}=\mathcal{T}^*$.
\end{defn}

For a diagonal Weyl channel given as mixture of unitaries, symmetry is equivalent to $p(\textbf{a},\textbf{b})=p(-\textbf{a},-\textbf{b})$, as can be readily checked. In particular, this implies that if the underlying systems are qubits, the resulting channel will always be symmetric, as $(\textbf{a},\textbf{b})=(-\textbf{a},-\textbf{b})$. Moreover, many relevant noise models, such as depolarizing and dephasing channels satisfy this assumption.

One simple corollary of this property is that all eigenvalues of the channel are real because it is a symmetric operator with respect to the Hilbert-Schmidt scalar product. This property makes the task of estimating the phases significantly easier: with the assumption that the noise model is symmetric and Weyl diagonal, we have
\begin{align}
\mu(\textbf{a},\textbf{b})=d^{-n}\lambda(\textbf{a},\textbf{b})\tr{ W_{(\textbf{a},\textbf{b})}^\dagger  \mathcal{U} \lb W_{(\textbf{a},\textbf{b})}\rb},
\end{align}
where $\lambda(\textbf{a},\textbf{b})$ is a real number. If $\theta'$ is the known phase of $\tr{ W_{(\textbf{a},\textbf{b})}^\dagger  \mathcal{U} \lb W_{(\textbf{a},\textbf{b})}\rb}$, correctly identifying the phase of $\mu(\textbf{a},\textbf{b})$ boils down to determining if $\theta=\theta'+\pi$ or $\theta=\theta'$. This can be done by examining the signs of $Y_l$ and $Z_l$. 

Thus, we conclude that in the symmetric case we can estimate all diagonal entries with multiplicative precision.

\subsection{Stability of the linear fitting}
In the previous section we showed how to obtain a multiplicative estimate on the diagonal entries of the noisy unitaries.
However, in order to get an efficient description of noise we then need to fit these diagonal elements to a noise model. Let $\hat{\mu}\in \C^{m}$ be the vector with our (noisy) estimates from $m$ different randomized benchmarking experiments, $\mu$ be the true values and $W_{(\textbf{a}_1,\textbf{b}_1)},\ldots,W_{(\textbf{a}_m,\textbf{b}_m)}$ be the Weyl operators corresponding to the data.
We want to fit this data to a noise model given by a hypergraph $G=(V,E)$ describing the noise structure. We thus need to solve the system of linear equations given by:
\begin{align}\label{equ:linearequations}
d^{-n}\lb\sum\limits_{e\in E}f_e(\textbf{a}_i,\textbf{b}_i)\rb \tr{ W_{(\textbf{a},\textbf{b})}^\dagger  \mathcal{U} \lb W_{(\textbf{a}_i,\textbf{b}_i)}\rb}=\hat{\mu}(\textbf{a},\textbf{b}),
\end{align}
for $1\leq i\leq m$ in order to learn the parameters of $f_e$.
Let $A$ be the matrix that describes the linear system of equations from~\eqref{equ:linearequations}. 
\begin{thm}
	Let $A$ be the matrix defined above. Let $\mu,\hat{\mu}\in\C^{m}$ be the values and estimates of the experiment, respectively. Suppose that they satisfy:
	\begin{align}
	\|\mu-\hat{\mu}\|_{\ell_\infty}=\cO(\epsilon|1-\|\mu\|_{\ell_\infty}|^2)
	\end{align}
	for some $\epsilon>0$. Let $\hat{f}$ be the resulting vector with parameters after solving the linear regression problem:
	\begin{align}
	\min\limits_{f}\|Af-\hat{\mu}\|_{2},
	\end{align}
	where $f$ the true value of the parameters.
	Then:
	\begin{align}
	\|f-\hat{f}\|_{\infty}=\cO(\epsilon|1-\|\mu\|_{\infty}|^2\|(A^\dagger A)^{-1}A\|_{\infty\to\infty}),
	\end{align}
	where $\|\cdot\|_{\infty\to\infty}$ is the maximum sum of the absolute value of entries of a column.
	
\end{thm}
\begin{proof}
	We may write:
	\begin{align}
	\hat{f}=(A^\dagger A)^{-1}A \hat{\mu},\quad f=(A^\dagger A)^{-1}A \mu.
	\end{align}
	Therefore,
	\begin{align}
	&\|f-\hat{f}\|_{\infty}= \|(A^\dagger A)^{-1}A(\mu-\hat{\mu})\|_{\infty}\leq\\
	&\|A^\dagger A)^{-1}A\|_{\infty\to\infty} \|\mu-\hat{\mu}\|_{\infty}\leq\\
	&\cO( \epsilon\|A^\dagger A)^{-1}A\|_{\infty\to\infty}|1-\|\mu\|_{\infty}|^2), 
	\end{align}
	where in the last step we used our assumption on $\|\mu-\hat{\mu}\|_{\infty}$.
\end{proof}
Thus, given a noise model, we are able to determine the matrix $A$ and compute what is the required precision to obtain an estimate on the parameters of the noise. 
With this we complete the statistical and stability analysis of our randomized benchmarking protocol.

\section{Simulating noisy VQE}\label{sec:noisyvqe}
The variational quantum eigensolver or the quantum approximate optimization algorithm are two examples of hybrid quantum algorithms that have the potential to surpass classical methods when solving optimization problems~\cite{moll2017quantum} on near-term quantum hardware. Quantifying how noise affects the complexity of classically simulating a noisy quantum computer running VQE provides a valuable benchmark for validation and verification of these algorithms.

We will now show how one may use the results of our randomized benchmarking experiment for this purpose to devise sampling algorithms in the Weyl basis. More specifically, we will show that if noise is sufficiently local and we learned its classical description, then it is possible to upper bound the classical complexity of estimating local expectation values of outputs of the circuit. This is one of the key tasks accomplished on a quantum computer which runs VQE-like algorithms and this bound indicates when our classical simulation methods are efficient in this case.

The VQE can be broadly described as follows: given a Hamiltonian $H=\sum_iH_i$ on $n$ qudits such that each $H_i$ consists of tensor products of local observables, the goal is to approximate the ground state of this Hamiltonian. This is done by starting with a fixed state, say $\ketbra{0}{0}^{\otimes n}$, and applying a (local) circuits of depth $m$ to the state. The quantum computer is used to execute the transformation and subsequently measure the energy of the current state by estimating the expectation value of the corresponding local observables. This information is then used to update the circuit in order to generate the state that will be in lower energy space.
Our algorithm is well-suited for simulating this task by identifying practical regimes when samples can be generated efficiently. In general, depending on the strength (and locality) of the noise and the circuit considered, the number of samples required to obtain constant precision is exponential in $n$. We will focus on simulating the noisy quantum computer in the Heisenberg picture, as this will give a better scaling of the sample complexity. 
First, we need to find sampling oracles for the initial observables, for the quantum channels describing the noisy evolution and for the entries of the initial state in the Weyl basis. To estimate the classical sampling complexity, we then need to compute the relevant norms.

We will make use of gauge freedom in the Weyl representation of the state and of the operators. More precisely, define the vector representation of a local observable $O$ to be 
\begin{align}
O(\textbf{a},\textbf{b})=d^{-n}\tr{W_{(\textbf{a},\textbf{b})}^\dagger O},
\end{align}
and of the state to be 
\begin{align}
\rho(\textbf{a},\textbf{b})=\tr{W_{(\textbf{a},\textbf{b})}^\dagger \rho}.
\end{align}
We moved the $d^{-n/2}$ prefactor from the representation of $\rho$ into $O$. This is done to get a natural scaling for the relevant norms.

We start by estimating the relevant quantities for local observables:
\begin{lem}[Sampling oracles for local observables]\label{lem:locobseasy}
	Let $O=O_1\otimes O_2\otimes\cdots\otimes O_{n/k}$ be a product of $k$-local observables on $n$ qudits, where we assume for convenience that $k$ divides $n$. 
	Define
	\begin{align}
	O(\textbf{a},\textbf{b})=d^{-n}\tr{W_{(\textbf{a},\textbf{b})}^\dagger O}.
	\end{align}
	Then we can obtain a $\ell_1$ sampling oracle with respect to $O(\textbf{a},\textbf{b})$ in time $\cO(n k^{-1}d^{4k})$ and:
	\begin{align}\label{equ:l1normobservable}
	\|O\|_{\ell_1}=\prod\limits_{i=0}^{n/k} \|O_i\|_{\ell_1}.
	\end{align}
\end{lem}
\begin{proof}
	Note that $O$ is a tensor product of $k$ local observables, and we can compute the linear map that changes basis from the matrix entries to Weyl in time $\cO(d^{6k})$. Applying this map to each of the $O_i$ to compute their representation in the Weyl basis takes time $\cO(d^4)$. There are $n/k$ operators which brings the total time to $\cO(n k^{-1}d^{4k})$.
	Now note that the resulting vector in the Weyl basis is still a tensor product of $k$ vectors and we can obtain sampling oracles for each one of the in time $\cO(d^{2k})$. We then obtain a sampling oracle for $O(\textbf{a},\textbf{b})$ by taking independent samples of each of the product vectors.
	Eq.~\eqref{equ:l1normobservable} also follows from the observation that each one of the observables will still be of product form.
\end{proof}
The norm Eq.~\eqref{equ:l1normobservable} can scale exponentially with the number of qubits. But in many physically relevant scenarios it is $\cO(1)$. Examples include Pauli string observables on qubits and physically local observables, that is, those that only differ from the identity at a fixed number of sites.
To see the latter, note that the Pauli matrices are themselves part of the basis and for a Pauli observable $P$ we have with our choice of normalization that $\|P\|_{\ell_1}=1$.
For the physically local observables, $\|\one\|_{\ell_1}=1$, thus only a small number of terms in Eq.~\eqref{equ:l1normobservable} will be different from $1$.

The following lemma shows how to obtain oracles for the initial state. In this case, we need oracles for the entries of the state with respect to the Weyl basis.

\begin{lem}[Representations of states are bounded in the Weyl basis]\label{lem:prodstateeasy}
	Let $\rho\in\M_{d^n}$ be a product state on $n$ qudits. Define $\hat{\rho}$ as 
	\begin{align}\label{equ:entriesbounded}
	\rho(\textbf{a},\textbf{b})=\tr{W_{(\textbf{a},\textbf{b})}^\dagger\rho}.
	\end{align}
	Then $\|\rho(\textbf{a},\textbf{b})\|_{\ell_\infty}\leq 1$ and, given $(\textbf{a},\textbf{b})$, we can compute $\rho(\textbf{a},\textbf{b})$ in time $\cO(nd^2)$.
\end{lem}
\begin{proof}
	We apply H\"older's inequality 
	\begin{align}
	\tr{W_{(\textbf{a},\textbf{b})}^\dagger\rho}\leq \|W_{(\textbf{a},\textbf{b})}\|_{\infty}\|\rho\|_1=1,
	\end{align}
	which gives~\eqref{equ:entriesbounded}. To see the complexity of computing an entry, note that as $\rho$ is assumed to be product and $W_{(\textbf{a},\textbf{b})}$ is product as well, the trace factorizes and we only need to compute $n$ traces of the product of $d\times d$ matrices.
\end{proof}
Note that variations of the statements above also hold for other choices of local, product bases for the set of matrices.

Using these estimates it is then easy to use information about the noise to upper-bound the complexity needed to estimate the expectation of a local observable on a quantum circuit. 
For instance, let us assume that our circuit consists of a sequence of $t$ gates $U_i$ acting on at most $2$-qudits followed by Weyl-diagonal noise $\mathcal{T}_i$ acting on the same qudits as the gate. Furthermore, we assume that the initial state is product and we measure a Pauli string observable.
Using our randomized benchmarking procedure we can efficiently learn the diagonals of $\mathcal{T}_i$. Moreover, as we showed in Sec.~\ref{app:statstability} of the Supplemental Material,  learning the absolute value of the diagonal elements is particularly efficient. This information is sufficient to estimate $\|\mathcal{T}_i\circ\mathcal{U}_i\|_{\ell_1\to\ell_1}$. If we only have $2-$local (noisy) gates, we can estimate this norm efficiently and
\begin{align}
\prod\limits_{i=i}^{t}\|\mathcal{T}_i\circ\mathcal{U}_i\|_{\ell_1\to\ell_1}^2
\end{align}
is the required classical overhead for the number of samples required to estimate the expectation value.

To illustrate this, consider circuits consisting of noisy Clifford and noisy T gates:
\begin{prop}[Local observables in noisy Clifford+T circuits]
	Consider a circuit on $n$ qudits consisting of $n_C$ two qudit Clifford gates followed by a two-local depolarizing noise with depolarizing parameter $p_C$ and $n_T$ $T$ gates followed by one-local depolarizing noise with parameter $p_T$. Suppose that the initial state is a product state and denote by $\rho$ the output state of the circuit. Let $O$ be an observable supported on $k=\cO(1)$ qudits or a Pauli string observable.
	Then, with  probability of success at least $1-\delta$, we can estimate $\tr{\rho O}$ up to an additive error $\epsilon>0$ in time $\cO(\operatorname{poly}(n)\lb p_T\sqrt{2}\rb^{2n_T}p_C^{2n_C}\epsilon^{-2}\log(\delta^{-1}))$. 
\end{prop}
\begin{proof}
	Using lemmas~\ref{lem:locobseasy} and~\ref{lem:prodstateeasy}, we know that we can get sampling oracles for $O$, and the initial state takes $\textrm{poly}(n)$ time to generate a sample. Given that the gates in the circuit are local, we can also sample from the intermediate steps in time $\textrm{poly}(n)$. Threfore, the complexity of generating a sample is $\textrm{poly}(n)$. By the above lemmas simulating the circuit in the Heisenberg picture we get $\|\rho\|_{\ell_\infty} \|O\|_{\ell_1}=\cO(1)$. It only remains to estimate the negativity generated by the gates to obtain the finite estimate. As discussed in the main text, this is at most $\lb p_T\sqrt{2}\rb^{n_T}p_C^{n_C}$. Thus, $\cO(\textrm{poly}(n)\lb p_T\sqrt{2}\rb^{2n_T}p_C^{2n_C}\epsilon^{-2})$ many samples suffice to obtain an estimate with the required precision.
\end{proof}
It is possible to generalize the statement of the Proposition by considering the same bound in the phase space basis or considering other gates and noise models. 

Thus, we conclude that our protocol gives rise to an efficient way of measuring the power of the quantum computer with a clear operational interpretation: it gives an upper bound on the complexity of classical simulation circuits used for the VQE. 

\section{Bases for sampling}\label{app:bases}

We will now show that we may sample and estimate $N_B$ from Theorem 2 in the main text efficiently for a range of product bases. Given an orthonormal set of matrices with respect to the Hilbert-Schmidt scalar product $\{B_l\}_{l=1}^{d^2}$ of $\M_d$, we may define an orthonormal basis of $\M_{d^n}$ by just taking tensor products of the basis elements. We will call a basis of $\M_{d^n}$ a \emph{product basis} if it is of this form. Here are some examples:
\begin{example}[Standard basis]
	One example of a product basis of $\M_{d^n}$ is $\{\ketbra{i}{j}\}_{i,j=1}^d$, where $\ket{i},\ket{j}$ are just elements of the computational basis. This basis is a good choice if operations in the circuit are dominated by measurements in the computational basis, and the state preparation is adaptive, i.e., it is close to a classical Markov chain. It is also a natural choice when simulating the evolution of sparse Hamiltonians.
\end{example}

\begin{example}[Weyl basis]
	For prime values of $d$, another useful basis is that given by the normalized Weyl unitaries $\sqrt{d}^{-1}W_{(a,b)}$. 
	As we saw before, many noise models are diagonal and have a particularly simple description in this basis. Moreover, this basis is also a good choice for circuits that are dominated by Clifford gates. To see why this is the case, let $\mathcal{U}_C$ be the conjugation with a Clifford unitary $C$. By definition, $C$ is an element of the normalizer of the Weyl group. The matrix $\hat{T}_C$ is a monomial unitary matrix in the Weyl representation and, thus, that $\|\hat{T}_C\|_{\ell_1\to\ell_1}=1$. 
\end{example}

\begin{example}[Phase space basis]
	Another important example is given by the phase basis~\cite{Gross2006,Gross2007,1903.04483,Veitch2012,Heinrich2019}. It gives another choice of basis for which Clifford circuit elements can be simulated efficiently and has an extra feature that states are quasiprobability distributions in it, that is, we have the extra property that $\sum_i\hat{\rho}(i)=1$. The matrix $\hat{\rho}$ is a representation of the operator $\rho$ in the phase space basis.
	
\end{example}
There are two main features desirable from a 'good' basis: (a) it is possible to obtain samples efficiently, (2) the constant $M_B$ (which is referred to as negativity in the discrete phase space literature~\cite{Gross2006,Gross2007,1903.04483,Veitch2012,Heinrich2019}) is small.

\section{Lindbladian evolution and computing matrix exponentials}
\label{app:matrixexpo}
There are scenarios when it is more natural to express the evolution in continuous time, as opposed to using the circuit model. Such evolution is described by a Lindbladian. We now show how to adapt our framework to this setting.

As mentioned in the main text, our method is also suited to compute exponentials of matrices for short evolution times. This extends our methods to simulating sparse Lindbladians, which again encompass both Hamiltonian dynamics and dissipative evolutions.
Here we show how to compute exponentials of Lindbladians based on our algorithm. 

Consider the Lindbladian $\mathcal{L}\in\M_d$ with operator norm $\|\mathcal{L}\|\leq1$. Assume that we have access to $\ell_1$ samples of rows. That is, given some row $i$ of $\mathcal{L}^{(k)}$, we can draw samples from the distribution of the entries given by
\begin{align}
p_k(j|i)=\frac{|\mathcal{L}^{(k)}(i,j)|}{\|\hat{\mathcal{L}}^{(k)}(i)\|_{\ell_1}}.
\end{align}
Where $\hat{\mathcal{L}}^{(k)}$ is the representation of $\mathcal{L}^{(k)}$ in some product basis.
We showed earlier that this can be done efficiently if we we impose locality constraints on $\mathcal{L}$. Also note that this can be done efficiently if we have the promise that each row of $\hat{\mathcal{L}}$ contains only $s$ nonzero entries and we are in the sparse input model, i.e., for each row $i$ we are given a list of the indices of the $s$ nonzero entries. 
Denoting by $\hat{\rho}$ the representation of the operator $\rho$ in product basis the same as for $\hat{\mathcal{L}}$,  the algorithm to compute $\tr{e^{t_n\mathcal{L}^{(n)}}\circ\cdots\circ e^{t_1\mathcal{L}^{(1)}}(\rho)E}$ is as follows:
\begin{tcolorbox}[breakable]
\emph{Input:} noisy quantum circuit specified by Linbladians $\mathcal{L}^{(1)},\ldots,\mathcal{L}^{(n)}$ and times $t_1,\ldots,t_n$, initial quantum state $\rho$ and observable $E$.\\
\emph{Output:} complex number $y$ s.t. $\mathbb{E}(y)=\tr{e^{t_n\mathcal{L}^{(n)}}\circ\cdots\circ e^{t_1\mathcal{L}^{(1)}}(\rho)E}$
	\begin{enumerate}
		\item Sample $i_0$ from the distribution $p_0(i_0)=\frac{|\rho(i_0)|}{\|\hat{\rho}\|_{\ell_1}}$.
		
		\item For $l=1,\ldots,n$:
		\begin{enumerate}
		\itemsep-0.2em  
			\item Draw $q_l$ from a Poisson distribution with parameter $t_l$.
			\item Set $s_{0,l}=i_l$
			\item For $m=1,\ldots,q_l-1$:
			\begin{itemize}
				\item Sample $s_{m+1,l}$ from $p_l(\cdot|s_{m,l})$.
			\end{itemize}
		\end{enumerate}

		\item Output $y$ given by
		\begin{align*}
		&y=e^{t}\operatorname{sign}(\rho(i_0))\|\hat{\rho}\|_{\ell_1}\times\\
		&\prod\limits_{l=1}^{n}\prod\limits_{m=1}^{k_l}\|\mathcal{L}^{(j)}(s_{m,l})\|_{\ell_1} \operatorname{sign}(\mathcal{L}^{(j)}(s_{m,l},s_{m+1,l}))\hat{E}(i_{n}),
		\end{align*}
		where $t=\sum_it_i$.
	\end{enumerate}
\end{tcolorbox}
We then have:
\begin{thm}
	The expectation value of the output of the algorithm above is
	\begin{align}
	\tr{e^{t_n\mathcal{L}^{(n)}}\circ\cdots\circ e^{t_1\mathcal{L}^{(1)}}(\rho)E}.
	\end{align}
	Its variance $\sigma^2$ satisfies:
	\begin{align}
	\sigma^2\leq \operatorname{exp}\left( \sum\limits_{i=1}^nt_i(\|\mathcal{L}^{(i)}\|_{\ell_1\to \ell_1}^2+1)\right)\|\rho\|_{\ell_1}^2\|\hat{E}\|_{\ell_\infty}^2
	\end{align}
\end{thm}
\begin{proof}
	Note that by conditioning on the values of $q_1,\ldots,q_n$ of the Poisson random vzriables we see that the algorithm above coincides with the one we described for classical circuits before with the sequence of evolutions given by 
	\begin{align}
	\left(\mathcal{L}^{(n)}\right)^{q_n}\circ\cdots\circ\left(\mathcal{L}^{(1)}\right)^{q_1}.
	\end{align}
	Thus, the expectation value conditioned on $q_1,\ldots,q_l$ is 
	\begin{align}
	e^{t}\tr{\left(\mathcal{L}^{(n)}\right)^{q_n}\circ\cdots\circ\left(\mathcal{L}^{(1)}\right)^{q_1}(\rho)E}.
	\end{align}
	The probability of observing each outcome is:
	\begin{align}
	\prod\limits_{i=1}^n\frac{e^{-t_i}}{q_i!}.
	\end{align}
	Thus, the expectation value of the output is:
	\begin{equation}
	\begin{split}
	&e^{t}\sum\limits_{q_1,\ldots,q_n=0}^{\infty}\prod\limits_{i=1}^n\frac{e^{-t_i}}{k_i!}\tr{\left(\mathcal{L}^{(n)}\right)^{q_n}\circ\cdots\circ\left(\mathcal{L}^{(1)}\right)^{q_1}(\rho)E}=\\
	&\tr{e^{t_n\mathcal{L}^{(n)}}\circ\cdots\circ e^{t_1\mathcal{L}^{(1)}}(\rho)E},
	\end{split}
	\end{equation}
	where the last equality follows from the Taylor expansion of the exponential function.
	It now remains to bound the variance of the output.
	Conditioned on $q_1,\ldots,q_n$, the output of the algorithm is bounded by 
	\begin{equation}
	\begin{split}
	e^{t}\|E\|_\infty\|\rho\|_{\ell_1}&\prod_{i=1}^n\|\left(\mathcal{L}^{(i)}\right)^{q_i}\|_{\ell_1\to \ell_1}\leq \\
	&e^{t}\|E\|_\infty\|\rho\|_{\ell_1}\prod_{i=1}^n\|\mathcal{L}^{(i)}\|_{\ell_1\to \ell_1}^{q_i}.
	\end{split}
	\end{equation}
	We bound the second moment of the output by:
	\begin{equation}
	\begin{split}
	\sum\limits_{k_1,\ldots,k_n=0}^{\infty}\prod\limits_{i=1}^n\frac{e^{-t_i}}{k_i!}\lb e^{t}\|E\|_\infty\|\rho\|_{\ell_1}\prod_{i=1}^n\|\mathcal{L}^{(i)}\|_{\ell_1\to \ell_1}^{k_i}\rb^2=\\
	e^{2t}\|E\|_\infty^2\|\rho\|_{\ell_1}^2\prod\limits_{i=1}^n\lb \sum\limits_{k=0}^\infty e^{-t_i}e^{2\log\|\mathcal{L}^(i)\|_{\ell_1\to \ell_1}k}\frac{t^k}{k!}\rb.
	\end{split}
	\end{equation}
	Now note that 
	\begin{align}
	\lb \sum\limits_{k=0}^\infty e^{-t_i}e^{2\log\|\mathcal{L}\|_{\ell_1\to \ell_1}k}\frac{t^k}{k!}\rb=\mathbb{E}\lb e^{2\log(\|\mathcal{L}\|_{1\to 1}), k}\rb,
	\end{align}
	where we are taking the expectation value with respect to a Poisson distribution with parameters $t_i$. This is just the moment generating function of the Poisson distribution with parameter $t_1$ at $2\log(\|\hat{\mathcal{L}}^{(i)}\|_{1\to 1})$.  Thus:
	\begin{align}
	\lb \sum\limits_{k=0}^\infty e^{-t_i}e^{2\log\|\hat{\mathcal{L}}^{(i)}\|_{\ell_1\to \ell_1}k}\frac{t^k}{k!}\rb=e^{t_i(2\|\hat{\mathcal{L}}^{(i)}\|_{1\to 1})},
	\end{align}
	where we used that the moment generating function of the Poisson distribution with parameter $t$ at $c$ is given by $\textrm{exp}\lb t (e^{c}-1)\rb$. 
	We conclude that the variance is bounded  by:
	\begin{align}
	\sigma^2\leq \operatorname{exp}\left(\sum\limits_{i=1}^nt_i(\|\mathcal{L}^{(i)}\|_{\ell_1\to \ell_1}^2+1)\right)\|\hat{\rho}\|_{\ell_1}\|\hat{E}\|_{\ell_\infty}.
	\end{align} 
\end{proof}
It  follows from Chebyshev's inequality that 
\begin{align}
\operatorname{exp}\left(\sum\limits_{i=1}^nt_i(\|\mathcal{L}^{(i)}\|_{\ell_1\to \ell_1}^2+1)\right)\|\hat{\rho}\|_{\ell_1}^2\|\hat{E}\|_{\ell_\infty}^2
\end{align}
samples suffice to estimate the scalar product up to and additive error $\epsilon$ with constant probability of success. Putting everything together, we get: 
\begin{thm}
	Let $0\leq t_1,\dots,t_n$, $\mathcal{L}^{(1)},\ldots,\mathcal{L}^{(n)}\in\M_d$ be a sequence of Lindbladians, $\rho$ a state and $E$ a POVM element. Assuming $\ell_1$ sampling for $\mathcal{L}^{(1)},\ldots,\mathcal{L}^{(n)}\in\M_d$ and $\rho$ we can estimate 
	\begin{equation}
	\tr{e^{t_n\mathcal{L}^{(n)}}\circ\cdots\circ e^{t_1\mathcal{L}^{(1)}}(\rho)E}
	\end{equation}
	up to an error $\epsilon>0$ with probability of success at least $2/3$ in expected time
	\begin{align}
	&\cO\|\rho\|_{\ell_1}\|E\|_{\ell_\infty} t_{\textrm{tot}}\times \\& \prod\limits_{i=1}^{n}\operatorname{exp}\left(\frac{1}{2}t_i(\|\mathcal{L}^{(i)}\|_{\ell_1\to\ell_1}^2+1)\right),
	\end{align} 
	where $t_{\textrm{tot}}=\sum\limits_{i=1}t_i$.
\end{thm}

Using the above, it is straightforward to adapt the remaining statements regarding classical simulability to continuous time evolutions.
\bibliography{biblio}
\end{document}